\documentclass[letterpaper,11pt]{article}
\pdfoutput=1
\usepackage{microtype}
\usepackage[letterpaper,margin=1in]{geometry}
\usepackage[UKenglish]{babel}
\usepackage[numbers,sort&compress]{natbib}
\usepackage{color}
\usepackage{doi}
\usepackage{latexsym}
\usepackage{amsmath}
\usepackage{amssymb}
\usepackage{amsthm}
\usepackage[all,warning]{onlyamsmath}
\usepackage{enumitem}
\usepackage{xparse}
\usepackage{mathtools}
\usepackage{graphicx}
\usepackage[basic]{complexity}
\usepackage[small]{caption}
\usepackage{wrapfig}
\usepackage{hyperref}

\theoremstyle{plain}
\newtheorem{theorem}{Theorem}
\newtheorem{lemma}[theorem]{Lemma}
\newtheorem{proposition}[theorem]{Proposition}

\newtheorem{claim}[theorem]{Claim}

\theoremstyle{definition}
\newtheorem{definition}[theorem]{Definition}

\theoremstyle{remark}

\setlist[itemize]{label=--}
\setlist[enumerate]{label=(\arabic*),labelindent=\parindent,leftmargin=*}

\DeclarePairedDelimiter\braces{\{}{\}}
\NewDocumentCommand\set{O{}mg}{\ensuremath{\braces[#1]{#2\IfNoValueTF{#3}{}{\,:\,#3}}}}

\DeclareMathOperator{\indeg}{in-deg}
\DeclareMathOperator{\outdeg}{out-deg}
\DeclareMathOperator{\dist}{dist}

\DeclareMathOperator{\mypoly}{poly}

\DeclareMathOperator{\diag}{diag}

\newcommand{\N}{\mathbb{N}}
\newcommand{\Z}{\mathbb{Z}}

\newclass{\lcl}{LCL}
\newclass{\local}{LOCAL}

\newcommand{\namedref}[2]{\hyperref[#2]{#1~\ref*{#2}}}
\newcommand{\sectionref}[1]{\namedref{Section}{#1}}

\newcommand{\theoremref}[1]{\namedref{Theorem}{#1}}

\newcommand{\figureref}[1]{\namedref{Figure}{#1}}
\newcommand{\lemmaref}[1]{\namedref{Lemma}{#1}}

\newenvironment{myabstract}
               {\list{}{\listparindent 1.5em%
                        \itemindent    \listparindent
                        \leftmargin    1cm
                        \rightmargin   1cm
                        \parsep        0pt}%
                \item\relax}
               {\endlist}

\newenvironment{mycover}
               {\list{}{\listparindent 0pt
                        \itemindent    \listparindent
                        \leftmargin    1cm
                        \rightmargin   1cm
                        \parsep        0pt}%
                \raggedright
                \item\relax}
               {\endlist}

\newcommand{\myemail}[1]{\,$\cdot$\, {\small #1}}
\newcommand{\myaff}[1]{\,$\cdot$\, {\small #1}\par\medskip}

\hypersetup{
  colorlinks=true,
  linkcolor=black,
  citecolor=black,
  filecolor=black,
  urlcolor=[rgb]{0,0.1,0.5},
  pdftitle={LCL problems on grids},
  pdfauthor={}
}

\begin{document}

\newgeometry{margin=1in,bottom=0.2in}

\begin{mycover}
{\huge\bfseries\boldmath \lcl{} problems on grids \par}
\bigskip

\textbf{Sebastian Brandt}
\myemail{sebastian.brandt@tik.ee.ethz.ch}
\myaff{ETH Z\"urich}

\textbf{Juho Hirvonen}
\myemail{juho.hirvonen@aalto.fi}
\myaff{IRIF, CNRS and University Paris Diderot}

\textbf{Janne H.\ Korhonen}
\myemail{janne.h.korhonen@aalto.fi}
\myaff{Aalto University}

\textbf{Tuomo Lempi\"ainen}
\myemail{tuomo.lempiainen@aalto.fi}
\myaff{Aalto University}

\textbf{Patric R.\ J.\ \"Osterg\aa{}rd}
\myemail{patric.ostergard@aalto.fi}
\myaff{Aalto University}

\textbf{Christopher Purcell}
\myemail{christopher.purcell@aalto.fi}
\myaff{Aalto University}

\textbf{Joel Rybicki}
\myemail{joel.rybicki@helsinki.fi}
\myaff{University of Helsinki}

\textbf{Jukka Suomela}
\myemail{jukka.suomela@aalto.fi}
\myaff{Aalto University}

\textbf{Przemys\l{}aw Uzna\'nski}
\myemail{przemyslaw.uznanski@inf.ethz.ch}
\myaff{ETH Z\"urich}
\end{mycover}

\medskip
\begin{myabstract}
\noindent\textbf{Abstract.}
\lcl{}s or locally checkable labelling problems (e.g.\ maximal independent set, maximal matching, and vertex colouring) in the \local{} model of computation are very well-understood in cycles (toroidal 1-dimensional grids): every problem has a complexity of $O(1)$, $\Theta(\log^* n)$, or $\Theta(n)$, and the design of optimal algorithms can be fully automated.

This work develops the complexity theory of \lcl{} problems for toroidal 2-dimensional grids. The complexity classes are the same as in the 1-dimensional case: $O(1)$, $\Theta(\log^* n)$, and $\Theta(n)$. However, given an \lcl{} problem it is undecidable whether its complexity is $\Theta(\log^* n)$ or $\Theta(n)$ in 2-dimensional grids.

Nevertheless, if we correctly guess that the complexity of a problem is $\Theta(\log^* n)$, we can completely automate the design of optimal algorithms. For any problem we can find an algorithm that is of a normal form $A' \circ S_k$, where $A'$ is a finite function, $S_k$ is an algorithm for finding a maximal independent set in $k$th power of the grid, and $k$ is a constant.

Finally, partially with the help of automated design tools, we classify the complexity of several concrete \lcl{} problems related to colourings and orientations. 
\end{myabstract}

\vfill
\noindent\makebox[\textwidth]{\includegraphics[width=8.1in]{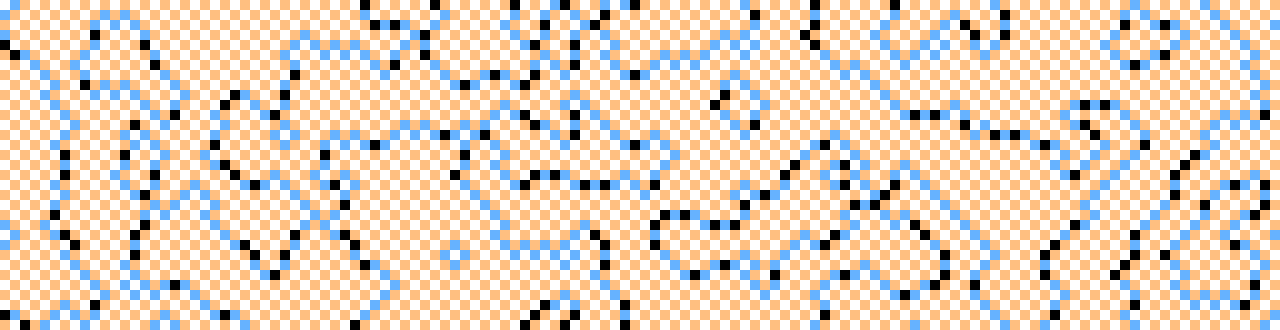}}

\thispagestyle{empty}
\setcounter{page}{0}
\newpage
\restoregeometry

\section{Introduction}\label{sec:introduction}

\subsection{Problem setting: \lcl{} problems on grids}

\begin{wrapfigure}{r}{50mm}
\vspace{-4mm}%
\hspace{1mm}%
\includegraphics[width=49mm,page=1]{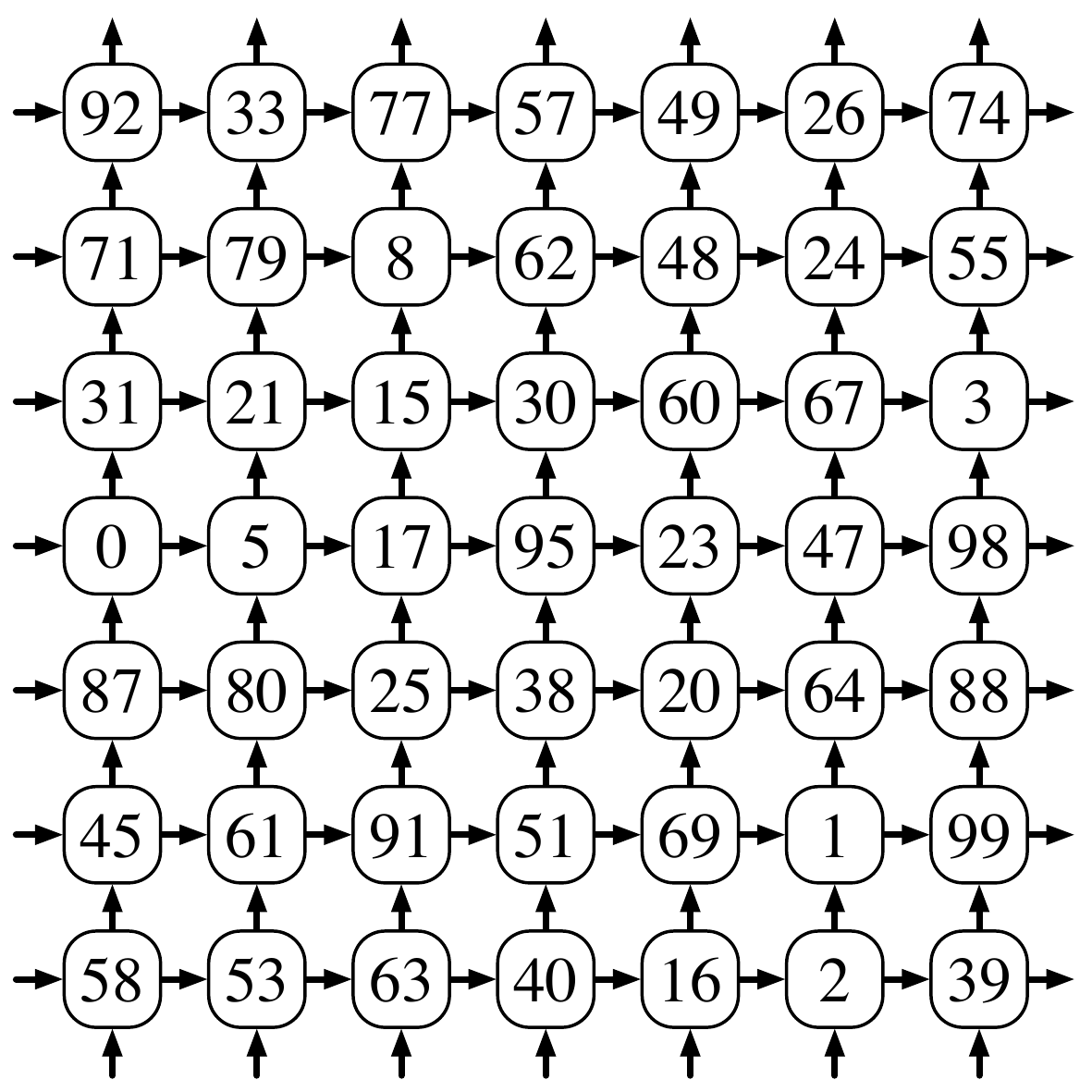}
\vspace{-10mm}%
\end{wrapfigure}

\paragraph{Grids.} In this work, we study distributed algorithms in a setting where the underlying input graph is a grid. Specifically, we consider the complexity of locally checkable labelling problems, or \lcl{} problems, in the standard \local{} model of distributed complexity, and consider graphs that are toroidal two-dimensional $n \times n$ grids with a consistent orientation; we focus on the two-dimensional case for concreteness, but most of our results generalise to $d$-dimensional grids of arbitrary dimensions.

This setting occupies a middle ground between the well-understood directed $n$-cycles~\cite{naor95what,chang16exponential}, where all solvable \lcl{} problems are known to have deterministic time complexity either $O(1)$, $\Theta(\log^* n)$ or $\Theta(n)$, and the more complicated setting of general $n$-vertex graphs, where intermediate problems with time complexities such as $\Theta(\log n)$ are known to exist, even for bounded-degree graphs. Grid-like systems with local dynamics also occur frequently in the study of real-world phenomena. However, grids have so far not been systematically studied from a distributed computing perspective.

\paragraph{\boldmath \local{} model and \lcl{} problems.} In the \local{} model of distributed computing, nodes are labelled with unique numerical identifiers with $O(\log n)$ bits. A time-$t$ algorithm in this model is simply a mapping from radius-$t$ neighbourhoods to local outputs; equivalently, it can be interpreted as a message-passing algorithm in which the nodes exchange messages for $t$ synchronous rounds and then announce their local outputs.

\lcl{} problems are graph problems for which the feasibility of a solution can be \emph{verified} by checking the solution for each $O(1)$-radius neighbourhood; if all local neighbourhoods look valid, the solution is also globally valid. Examples of such problems include vertex colouring, edge colouring, maximal independent sets, and maximal matchings. We refer to \sectionref{sec:prelim} for precise definitions.

\paragraph{Example: colouring the grid.} To illustrate the type of questions we are interested in this work, consider $k$-colouring on $n \times n$ grids. For $k = 2$, the problem is inherently global with complexity $\Theta(n)$, while colouring any graph of maximum degree $\Delta=4$ with $\Delta + 1 = 5$ colours can be done in $O(\log^* n)$ rounds. But what about $k = 3$ and $k = 4$? In particular, does either of these have an intermediate (polylogarithmic) complexity, as is known to happen with $\Delta$-colouring on general bounded-degree graphs~\cite{panconesi95local,chang16exponential}? We will see that neither $3$-colouring nor $4$-colouring is intermediate on grids: $3$-colouring requires $\Theta(n)$ rounds, while $4$-colouring can be solved in $O(\log^* n)$ rounds.

\subsection{Results: classification and synthesis}

\paragraph{Classification.}
Our first contribution is a complete \emph{complexity classification} for \lcl{} problems on grids in the case of deterministic algorithms. That is, we show that any \lcl{} problem on $n \times n$ grids has one of the following time complexities, similarly to the case of cycles:
\begin{enumerate}[label=(\alph*),noitemsep]
\item $O(1)$ (``trivial'' problem)
\item $\Theta(\log^* n)$ (``local'' problem)
\item $\Theta(n)$ (``global'' problem)
\end{enumerate}
In particular, there are no problems of an intermediate complexity, such as $\Theta(\log n)$.

The separation between $O(1)$ and $\Omega(\log^* n)$ follows from the work of \citet{naor95what} (see Appendix~A of \citet{hierarchy}), and obviously all problems can be solved in $O(n)$ rounds on $n \times n$ grids (assuming they can be solved at all). The interesting part is the separation between (b) and (c); here we extend the recent speed-up lemma of \citet{chang16exponential} to grids.

\paragraph{Undecidability of classification.}
It is known that the classification for \lcl{} problems on cycles is decidable, that is, there is an algorithm that decides to which complexity class a given \lcl{} problem belongs (see \sectionref{sec:cycles}). We show that two-dimensional grids are fundamentally different from cycles in this regard: even if we have the promise that a given \lcl{} problem has complexity of either $\Theta(\log^* n)$ or $\Theta(n)$, distinguishing between these cases is undecidable.

\paragraph{\boldmath Algorithm synthesis for $\Theta(\log^* n)$ problems.} The undecidability result would seem to suggest that automating the design of distributed algorithms on $n \times n$ grids is essentially impossible. Surprisingly, this is not the case: we develop a \emph{synthesis algorithm} that, given a specification of an \lcl{} problem $P$ with complexity $O(\log^* n)$, produces an asymptotically optimal algorithm for $P$ on grids.  The caveat is that if the input problem $P$ is a global problem with complexity $\Theta(n)$, this algorithm cannot detect it and will never stop.

From a theory perspective, this means that for each \lcl{} problem $P$ we will only need 1 bit of advice---whether $P$ is $O(\log^* n)$ or $\Theta(n)$---and then we can find an optimal algorithm for solving~$P$: for $O(\log^* n)$ problems, we apply the synthesis algorithm, and for $\Theta(n)$ problems, brute force is optimal. From a practical perspective, we can use the synthesis algorithm as a one-sided oracle for understanding the complexity of \lcl{} problems on grids: if the synthesis produces an output, we have an optimal algorithm, and if it does not, we can conjecture that the problem in question might be inherently global.

\begin{figure}[b]
\centering
\includegraphics[width=\columnwidth,page=2]{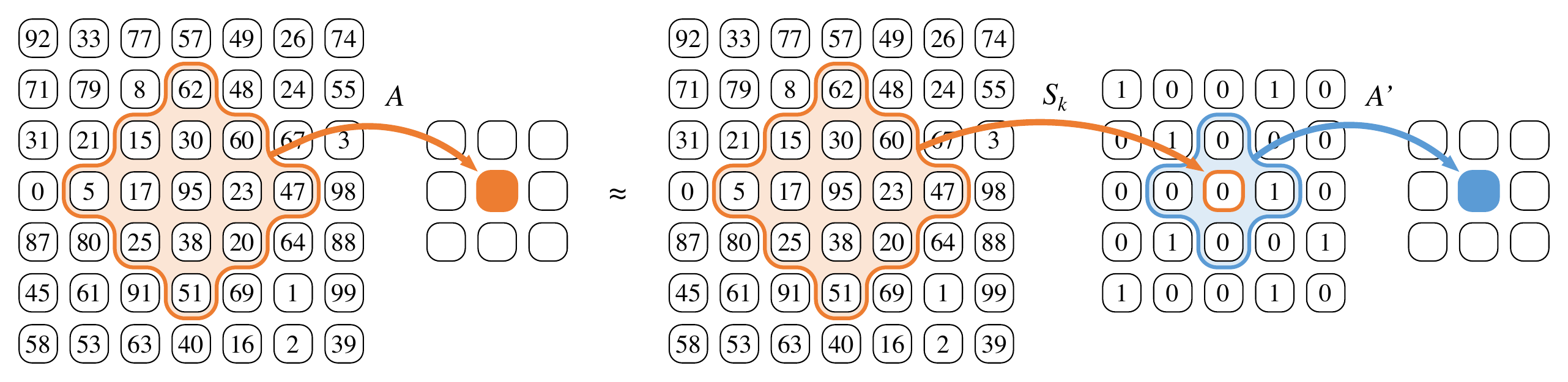}
\caption{Any sublinear-time algorithm $A$ can be normalised as $A' \circ S_k$, where $S_k$ is a problem-independent $O(\log^* n)$-time symmetry-breaking component and $A'$ is a problem-specific constant-time component.}\label{fig:normal-form}
\end{figure}

\paragraph{\boldmath Normal form for $\Theta(\log^* n)$ problems.} The algorithm synthesis is based on a result showing that every \lcl{} problem $P$ with complexity $\Theta(\log^* n)$ on $n \times n$ grids has an algorithm of a specific \emph{normal form}; see Figure~\ref{fig:normal-form}. That is, there is an algorithm $A$ for $P$ that has the form $A = A' \circ S_k$ for some constant $k$, where
\begin{itemize}
    \item $S_k$ is a problem-independent algorithm that finds a maximal independent set $I_k$ in the $k$th power of the $n \times n$ grid (we call these nodes ``anchors''), and
    \item $A'$ is a problem-dependent algorithm with running time $O(k)$ that takes as an input only the set of anchors $I_k$ and the global orientation of the grid.
\end{itemize}
Note that here only the part of finding the set of anchors takes $\Theta(\log^* n)$ time, and all of the remaining parts can be done in $O(1)$ time. In particular, the only problem-dependent part besides the constant $k$ is the finite function defining the algorithm $A'$; thus, the algorithm synthesis becomes a matter of searching through the finite-size space of possible functions.

\subsection{Results: upper and lower bounds for concrete \lcl{} problems}

Next, we turn our attention to concrete \lcl{} problems. In particular, we are interested in problem families defined for a range of parameters, so that it makes sense to ask where exactly is the border between local and global problems:
\begin{itemize}
  \item \emph{Vertex colouring.} The $k$-colouring problem is solvable in $O(\log^* n)$ rounds for $k \ge 4$ and global for $k \le 3$. 
  \item \emph{Edge colouring.} The $k$-edge colouring problem is solvable in $O(\log^* n)$ rounds for $k \ge 5$ and global for $k \le 4$.
  \item \emph{Edge orientations.} For a set $X \subseteq \{ 0,1,\dotsc,4 \}$, an $X$-orientation is an orientation of the edges such that for each node $v \in V$ we have $\indeg(v) \in X$. The problem is trivial if $2 \in X$. We show that if $\{0,1,3\} \subseteq X$ or $\{1,3,4\} \subseteq X$, the problem is solvable in $O(\log^* n)$ rounds, and otherwise it is global. 
\end{itemize}
The results on colourings can be generalised to $d$-dimensional grids. A $4$-colouring can be found in time $\Theta(\log^* n$) for any $d \ge 2$, while $3$-colouring is global. In the case of edge colouring, we show that a $(2d+1)$-colouring can be found in time $\Theta(\log^* n$) for any $d\ge 1$, while $2d$-colouring is global. Both of the upper bounds hold even without any orientation or dimensional information, while both of the lower bounds hold even with full information.

We remark that the techniques used in the vertex colouring results have been discovered before in the context of \emph{finitary colourings} of grids~\cite{finitarycol}; see \sectionref{sec:related-work} for more details.

\section{Related work}\label{sec:related-work}

\paragraph{\boldmath \lcl{} problems on cycles.} As we noted before, two-dimensional grids can be seen as a generalisation of the widely studied setting of cycles; indeed, \lcl{} problems were first studied on cycles in the distributed setting. Cole and Vishkin~\cite{cole86deterministic} showed that cycles can be 3-coloured in time $O(\log^* n)$, and Linial~\cite{linial92locality} showed that this is asymptotically optimal. This implies, via simple reductions, that many classical \lcl{} problems, such as maximal independent set and maximal matching, also have a complexity of $\Theta(\log^* n)$ on cycles.

\paragraph{\boldmath \lcl{} problems on graphs of bounded maximum degree.} Naor and Stockmeyer~\cite{naor95what} showed that there exists a non-trivial \lcl{} problem that can be solved in constant time: weak 2-colouring on graphs of odd degree. Many \lcl{} problems are known to either have complexity $\Theta(\log^* n)$~\cite{panconesi01some, barenboim14distributed, barenboim15sublinear, fraigniaud16local} or be global on graphs of bounded maximum degree. Until recently, no problems of an intermediate complexity were known. While Kuhn et al.~\cite{kuhn16local} gave a lower bound of $\min \{ \log \Delta / \log \log \Delta, \sqrt{\log n / \log \log n} \}$ for, among others, maximal independent set, this proof does not give an infinite family of graphs with a fixed maximum degree $\Delta$. Brandt et al.~\cite{brandt16lll} showed that sinkless orientation and $\Delta$-colouring have randomised complexity $\Omega(\log \log n)$, and Chang et al.~\cite{chang16exponential} proved that this implies a deterministic lower bound of $\Omega(\log n)$. These lower bounds provide the first examples of \lcl{} problems with provably intermediate time complexity. Ghaffari and Su~\cite{ghaffari17distributed} proved a matching upper bound for sinkless orientation; no tight bounds are known for $\Delta$-colouring, but there is a polylogarithmic upper bound due to Panconesi and Srinivasan~\cite{panconesi95local}.

\paragraph{\boldmath Complexity theory of \lcl{} problems.} \lcl{} problems were formally introduced by Naor and Stockmeyer~\cite{naor95what}. They showed that if there exists a constant-time algorithm for solving an \lcl{} problem $P$, then there exists an order-invariant constant-time algorithm for $P$, such that the algorithm only uses the relative order of unique identifiers given to the nodes. Their argument works for any time $t = o(\log^* n)$: a time-$t$ distributed algorithm implies a constant-time order-invariant algorithm; hence there are no \lcl{} problems with complexities strictly between $\omega(1)$ and $o(\log^* n)$.

Recently Chang et al.~\cite{chang16exponential} showed that there are further gaps in the time complexities of \lcl{} problems. They gave a speed-up lemma for simulating any deterministic $o(\log n)$-time algorithm in time $O(\log^* n)$ by computing new small and locally unique identifiers for the input graph. This implies that there are no \lcl{} problems with deterministic complexity $\omega(\log^* n)$ and $o(\log n)$. They also show that the deterministic complexity of an \lcl{} on instances of size $n$ is at most the randomised complexity on instances of size $2^{n^2}$. This implies a similar gap for randomised complexities between $\omega(\log^* n)$ and $o(\log \log n)$.

\paragraph{\boldmath \lcl{} problems in restricted graph families.} It appears that the complexity of \lcl{} problems specifically on grids has not been studied beyond the case of cycles. \lcl{} problems have been, however, studied on other restricted graph classes, such as graphs of bounded independence~\cite{kuhn05fast, gfeller07randomized, schneider10optimal, barenboim13edge-coloring}, bounded growth~\cite{schneider10symmetry} and bounded diversity~\cite{barenboim16deterministic}.

\paragraph{Existence of algorithms and algorithm synthesis.} The notion of automatic synthesis of algorithms has been around for a long time; for example, already in the 1950s Church proposed the idea of synthesising circuits~\cite{church57synthesis,vardi12church}. Since then synthesis of distributed and parallel protocols has become a well-established research area in the formal methods community~\cite{clarke82design,manna84synthesis,pnueli90distributed,attie01synthesis,kupferman01synthesizing,finkbeiner05uniform,bloem16synthesis}. However, synthesis has received considerably less attention in the distributed computing community, even though they have been used to discover e.g.\ novel synchronisation algorithms~\cite{bar-david03automatic,dolev16synthesis,bloem16synthesis} and local graph algorithms~\cite{rybicki15exact,hirvonen14local-maxcut}.

The synthesis of optimal distributed algorithms in general is often computationally hard and even undecidable. In the context of the \local{} model and \lcl{} problems, Naor and Stockmeyer~\cite{naor95what} show that simply deciding whether a given problem can be solved in constant time is undecidable; hence we cannot expect to completely automate the synthesis of asymptotically optimal distributed algorithms for \lcl{} problems in general graphs. This result holds even if we study non-toroidal two-dimensional grids, but it does not hold in toroidal grids. In essence, in toroidal grids only trivial problems are solvable in constant time, and as we will see, the interesting case is the time complexity of $O(\log^* n)$.

\paragraph{Other grid-like models.} While grids have not been studied from a distributed computing perspective, grid-like models with local dynamics have appeared in many different contexts:
\begin{itemize}
    \item \emph{Cellular automata}~\cite{neumann1966,gardner1970life,wolfram2002new} have been studied both as a primitive computational model, and as a model for various complex systems and emergent phenomena, e.g.\ in ecology, sociology and physics~\cite{kari2005theory,smith1976introduction,ganguly2003survey}.
    \item Various \emph{tiling models}~\cite{grunbaum87tilings} have connections to computability questions, such as the abstract \emph{Wang tilings}~\cite{wang61proving} and the variants of the \emph{abstract Tile Assembly Model (aTAM)}~\cite{winfree98algorithmic,doty12theory,patitz14introduction,woods15intrinsic} for DNA self-assembly.
\end{itemize}
However, the prior work of this flavour is usually interested in understanding the dynamics of a specific fixed process, or what kind of global behaviours can arise from fixed number of local states---in particular, whether the system is computationally universal. Our distributed complexity perspective to grid-like systems seems mostly novel, and we expect it to have implications in other fields. Applying an existing result of distributed computing to tiling models has been previously demonstrated by Sterling~\cite{sterling08limit}, who makes use of a weak-colouring lower bound by Naor and Stockmeyer~\cite{naor95what}.

\paragraph{Finitary colourings of grids.} Subsequently to the initial publication of this work, we have learned that the techniques in the $k$-colouring upper and lower bounds are essentially rediscoveries of prior work of \citet{finitarycol} in the context of \emph{finitary colourings of grids}. Very roughly speaking, this line of work concerns colouring the infinite $d$-dimensional grid $\mathbb{Z}^d$ using a specific type of random processes (\emph{factors}) with an independent and identically distributed random variable for each node; more generally, one can consider \emph{shifts of finite type}, which correspond to \lcl{} problems. In particular, \citet{finitarycol} study the \emph{coding radius} of factors, which is analogous to the running time of a distributed algorithm, and prove a separation between $3$-colouring and $4$-colouring. However, despite the fact that techniques seem to translate between finitary colourings and distributed complexity, it remains unclear how to directly translate results from one setting to the other in a black-box manner; for instance, can we derive the lower bound for $3$-colouring from the results of \citet{finitarycol}, and does our complexity classification imply answers to the open questions they pose?

\section{Preliminaries}\label{sec:prelim}

\paragraph{\boldmath \local{} model.}
In the \local{} model of distributed computing \cite{peleg00distributed,linial92locality}, we have a computer network that is represented as a graph $G = (V,E)$; each node is a computer and each edge is a bidirectional communication link. The computers collaborate in order to solve a graph problem related to the structure of the graph~$G$; note that here the same graph is both the topology of the computer network and the input graph.

Each node $v \in V$ is labelled with a unique identifier from the set $\{1,2,\dotsc,\mypoly(|V|)\}$. Each node has to produce its own part of the output: for example, if we are solving a graph colouring problem, each node has to output its own colour, and if our task is to find a maximal independent set $I$, each node has to output a binary label that indicates whether it is in set~$I$. This can be extended in a straightforward manner to edge labellings.

All nodes run the same deterministic algorithm. Computation proceeds in synchronous rounds. In each round, all nodes in parallel send messages (of an arbitrary size) to their neighbours, then the messages are propagated along the edges to the recipients, then all nodes in parallel receive messages from each of their neighbours, and finally all nodes update their local state. The running time of an algorithm is the number of communication rounds until all nodes have stopped and announced their local outputs.

Note that in a time-$t$ algorithm, each node can gather its radius-$t$ neighbourhood and choose its local output based on this information. In essence, a time-$t$ algorithm in the \local{} model is simply a mapping from radius-$t$ neighbourhoods to local outputs. Note that the neighbourhood contains not only the topology of the network but also the unique identifiers.

\paragraph{\boldmath \lcl{} problems.}

In distributed algorithms, the class of \lcl{} problems \cite{naor95what} plays a role somewhat analogous to the class \NP{} in centralised computing. Informally, problems in the class \lcl{} are precisely those problems that can be solved in \emph{constant time} with a \emph{nondeterministic} algorithm in the \local{} model: in \lcl{} problems all nodes can nondeterministically guess the solution and then verify it by checking that the solution looks consistent in all local neighbourhoods. The key question is which of the \lcl{} problems can be solved efficiently (e.g., in constant or near-constant time) with \emph{deterministic} algorithms. (Cf.: which problems in \NP{} are also in \P{}.)

More precisely, let $P$ be a graph problem that associates with each unlabelled input graph $G = (V,E)$ a set of feasible node labellings $P(G)$; here each $f\in P(G)$ is a mapping $f\colon V \to X$ for some set of output labels $X$. We say that $P$ is an \lcl{} problem if
\begin{enumerate}[noitemsep]
  \item the set of local outputs $X$ is a finite set of size $|X| = O(1)$,
  \item there is a constant $r = O(1)$ such that for any candidate labelling $f\colon V \to X$ we have $f\in P(G)$ if and only if each radius-$r$ neighbourhood is compatible with some $g\in P(G)$.
\end{enumerate}
Informally, for an \lcl{} problem, a solution is feasible if it looks like a feasible solution in all local neighbourhoods. Examples of such problems include $k$-vertex colouring for $k=O(1)$, maximal independent sets, and minimal dominating sets. Again, we can extend the definitions in a natural manner to edge labellings; hence also $k$-edge colourings for $k=O(1)$, maximal matchings, and minimal edge dominating sets can be interpreted as \lcl{} problems.

\paragraph{\boldmath Radius-$1$ \lcl{} problems.}
Above, parameter $r$ is called the \emph{checkability radius} or simply \emph{radius} of problem $P$. In bounded-degree graphs we can always define another \lcl{} problem $P'$ with radius $r' = 1$ that is equivalent to $P$ in the following sense: $P'$ can be solved in time $t$ in the \local{} model if and only if $P$ can be solved in time $t \pm O(1)$. In essence, the output labels in $P'$ are radius-$r$ neighbourhoods in $P$; given an algorithm for $P$ we can spend additional $r$ rounds to solve $P'$, and given an algorithm for $P'$, we can also directly solve $P$. Therefore we will often tacitly assume $r=1$, with the understanding that this will only influence additive constants in the running time.

\paragraph{Grid graphs.}
Unless otherwise stated, we will study graphs that are 2-dimensional toroidal $n \times n$ square grids. Define $G_n = (V_n, E_n)$, where $V_n = \{ (x,y) : 0 \le x,y < n \}$. We will use the shorthand $u = (x_u,y_u)$ for the coordinates of each node $u \in V_n$. The nodes do not have access to these coordinates. Two nodes $u$ and $v$ are connected by an edge if and only if $|x_u - x_v| + |y_u - y_v| = 1$, where all coordinates are modulo $n$. All edges are oriented in a consistent manner towards the larger coordinate, and labelled so that each node knows which edge points ``north'' (increasing $y$ coordinate), ``east'' (increasing $x$ coordinate), ``south'', and ``west''. By definition the grid wraps around in both dimensions, forming a torus. We will use the shorthands $V = V(G)$ and $E = E(G)$ for the node and edge sets, respectively, of $G$. We will assume that all nodes are given the value of $n$ as input.

\paragraph{On unsolvable problems.}

Many \lcl{} problems are unsolvable for some values of $n$. For example, there does not exist a $2$-colouring if $n$ is odd, and many problems are ill-defined for e.g.\ $n = 1$. Throughout this text we will usually assume that $n$ is sufficiently large so that the problem that we consider is meaningful. Problems for which there are infinitely many values of $n$ for which a solution does not exist (e.g.\ $2$-colouring) are regarded as global problems. Indeed, often the fact that solutions do not exist at all for some values of $n$ is a simple way of proving a lower bound of $\Omega(n)$, and such a bound holds even if we had a promise that $n$ is chosen so that a solution exists.

\paragraph{Notation.}

From now on, we write $G^{(k)}$ for the $k$th power of a graph~$G$. That is, $V(G^{(k)}) = V(G)$ and $E(G^{(k)}) = \{ \{u,v\} : \dist_G(u,v) \le k \}$. We denote the set $\{0,1,\ldots,k-1\}$ by $[k]$.

\section{Warm-up: directed cycles}\label{sec:cycles}

As a gentle introduction to our research questions, we will first have a look at \lcl{} problems in directed cycles (i.e., 1-dimensional grids). This case is completely understood by prior work, but we will present it from a new perspective: in the 1-dimensional case, any \lcl{} problem $P$ can be conveniently represented as a directed graph $H$. By studying elementary properties of graph $H$, we can directly deduce the computational complexity of problem $P$, and derive an asymptotically optimal algorithm for solving $P$---everything is decidable and algorithm synthesis is computationally tractable (see \figureref{fig:flexible}).

\begin{figure}
\centering
\includegraphics[page=3]{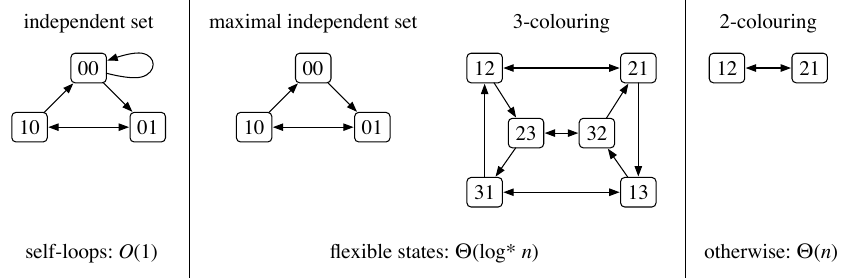}
\caption{\lcl{} problems for cycles can be represented as directed graphs. Here we have four \lcl{} problems with radius $r=1$; each node corresponds to a sequence of $2r=2$ output labels, and each edge corresponds to a sequence of $2r+1=3$ output labels. We can read the time complexity of the \lcl{} problem directly from the properties of the graph. For example, in the maximal independent set problem, state $00$ is flexible, as we can find walks from $00$ back to itself of lengths $3$ and $5$, and hence also of any length larger than $7$.}\label{fig:flexible}
\end{figure}

We construct an \emph{output neighbourhood graph} $H = (V,E)$ as follows. Problem $P$ can be interpreted as a set of feasible radius-$r$ local neighbourhoods $u_1 u_2 \dotsc u_{2r+1} \in P$. For every such neighbourhood, we will have an edge
$ ( u_1 u_2 \dotsc u_{2r},\, u_2 u_3 \dotsc u_{2r+1} ) \in E $
in graph $H$. For example, in the $3$-colouring problem, the sequence ``$132$'' is a feasible neighbourhood, and hence we will have an edge $(13, 32)$ in the graph; here e.g.\ ``$13$'' corresponds to a node with output $3$ that has a predecessor with output $1$ (see \figureref{fig:flexible}).
The key observation is that walks in graph $H$ correspond to feasible output labellings in problem $P$.

Now we say that a node $u \in V$ is {\em flexible} if there exists an integer $k$ such that for all $k'\geq k$, there exists a walk in $G$ of length $k'$ that starts and ends at $u$. We call the smallest such $k$ the {\em flexibility} of $u$. It is clear that $u$ is flexible if and only if there are circuits $C,C'$ containing $u$ whose respective lengths are coprime.

\begin{claim}
The complexity of $P$ is
$O(1)$ if some node of $H$ has a self-loop;
otherwise $\Theta(\log^* n)$ if some node of $H$ is flexible;
and otherwise $\Theta(n)$.
\end{claim}

\begin{proof}
The case of $O(1)$ time is straightforward. Recall the result of \citet{naor95what} that shows that unique identifiers do not help with $o(\log^* n)$-time algorithms; hence we have only trivial problems for which a constant labelling is a feasible.

Also the case of $\Theta(n)$ time is straightforward. There are only constantly many neighbourhoods, and hence some neighbourhood $u \in V$ has to be used $\Omega(n)$ times in the output. However, $u$ is not flexible, and hence the spacing between $u$-neighbourhoods requires global coordination. (For example, in $2$-colouring the distance between any two occurrences of neighbourhood ``$12$'' has to be a multiple of $2$.)

It remains to be shown that if $u$ is a flexible node with some minimum flexibility $k$, we can solve $P$ in time $O(\log^* n)$. Let $G$ be a cycle graph and let $G^{(k)}$ be the $k$th power of $G$. We can find a maximal independent set $I$ in $G^{(k)}$ in time $O(\log^* n)$. Let $v$ be a node in $I$, and let $v'$ be the next node in $I$ by the ordering of the nodes of $G$. Let the distance from $v$ to $v'$ in $G$ be $i$; we have $k+1 \leq i \leq 2k+1$. We label $v$ and $v'$ using the neighbourhood $u$, and fill in the gap between $v$ and $v'$ by following some circuit $C_i$ of length $i$ from $u$ back to $u$ in $H$.
\end{proof}

It would be tempting to try to generalise this result to $2$-dimensional grids. Unfortunately, this is not possible; as we will see in \sectionref{sec:undecidability}, there does not exist an algorithm for finding the time complexity of a given \lcl{} problem in $2$-dimensional grids. Nevertheless, we can still prove that any \lcl{} problem has a complexity of $O(1)$, $\Theta(\log^* n)$, or $\Theta(n)$ also in $2$-dimensional grids. We will next prove the key ingredient: any $o(n)$-time algorithm can be turned into an $O(\log^* n)$-time algorithm that has a convenient structure.

\section{Speed-up and normal form}\label{sec:normalform}

In this section, we give the speed-up result underlying both the complexity classification of \lcl{} problems and the synthesis. The following result is essentially a refined version of the speed-up lemma of \citet{chang16exponential} for two-dimensional oriented grids; the proof immediately yields the normal form algorithm for any \lcl{} problem as discussed in \sectionref{sec:introduction}.

\begin{theorem}\label{thm:speedup}
	Given any \lcl{} $P$ with an algorithm $A$ that solves $P$ in time $T(n) = o(n)$, there exists an algorithm $B$ that solves $P$ and has running time $O(\log^* n)$.
\end{theorem}

\begin{proof}
  Recall that w.l.o.g., we can assume that problem $P$ has checkability radius $r = 1$.
  Algorithm~$B$ solves problem $P$ in an $n \times n$ grid $G$ as follows:
	\begin{enumerate}[label=(\arabic*),noitemsep]
		\item Pick the smallest even $k \ge 4$ such that $T(k) < k/4-4$. Such a $k$ exists by assumption, and it is a constant that only depends on $T$.
		\item Find a maximal independent set $I$ in $G^{(k/2)}$. This can be done in time $O(\log^* n)$; the nodes of $I$ are called \emph{anchors}.
		\item Simulate $A$ with \emph{locally unique identifiers} from $[k^2]$ around each anchor in $I$. \label{sim:step3}
	\end{enumerate}

	Step \ref{sim:step3} of the simulation proceeds as follows. First, $G$ is divided into a Voronoi tiling with respect to $I$, breaking ties arbitrarily---that is, we associate with each node $v \in I$ a Voronoi tile
  $T(v) = \{ u \in V \colon v \text{ is the closest anchor to } u \}$;
  each node can compute which tile they belong to in constant time. Then, each node $v$ is assigned a \emph{local coordinate} $c(v) = (x_v - x_{a(v)}, y_v - y_{a(v)})$, where $a(v)$ is the anchor of $v$'s tile. The local coordinates will be interpreted as locally unique identifiers.

	There are no repeating identifiers within distance $k/2$ of any node: If two nodes $u$ and $v$ have the same coordinate, they are in different Voronoi tiles. Since the anchors are at distance at least $k/2$, and $u$ and $v$ are by assumption in the same relative positions with respect of their anchors, also $u$ and $v$ are at distance at least $k/2$.

	Each Voronoi tile $T(v)$ holds nodes at distance at most $k/2+1$, since any node at distance $k/2+2$ from $v$ must have another anchor within distance $k/2$. We can calculate that the size of each tile is at most $|T(v)| \le k^2$; hence we only need $k^2$ distinct locally unique identifiers.

	Next we simulate $A$ on $G$, with a bit of cheating: we tell $A$ that we are actually solving $P$ for an instance of size $k \times k$; for each local neighbourhood of $G$, we feed it locally unique identifiers from $[k^2]$. Now $A$ has a running time $T(k) < k/4$, and hence it does not ever see repeating identifiers; it has to solve problem $P$ correctly in each local neighbourhood as this might be a legitimate instance of size $k \times k$. More precisely, if the local outputs of $A$ violated the constraints of the \lcl{} problem $P$ for some local neighbourhoods, we could also construct a genuine instance $H$ of size $k \times k$ with globally unique identifiers, and $A$ would fail to solve $P$ in $H$. Hence the local outputs of $A$ have to constitute a globally feasible solution for $P$ also in $G$.
\end{proof}

\section{Undecidability of classification}\label{sec:undecidability}

In this section we show that in general deciding whether the running time of a given \lcl{} problem is $\Theta(\log^* n)$ or $\Theta(n)$ is undecidable. We achieve this by defining, for each Turing machine $M$, an \lcl{} problem $L_M$ such that $L_M$ can be solved in time $\Theta(\log^* n)$ if and only if $M$ halts, and in time $\Theta(n)$ otherwise.

\begin{theorem} \label{thm:lcl-undecidability}
	The problem of deciding whether a given \lcl{} can be solved in time $\Theta(\log^* n)$ or $\Theta(n)$ on grids is undecidable.
\end{theorem}

It is good to compare this with the result of \citet{naor95what}. They study grids with boundaries (non-toroidal grids; there are nodes of degrees $3$ and $2$). In such grids, deciding if an \lcl{} can be solved in time $O(1)$ is already undecidable. In essence, for any Turing machine $M$ we can construct an \lcl{} that specifies that in the lower-left corner of the grid we have to write out the complete execution history of $M$, and everything else can be padding. Now if and only if $M$ halts in some finite time $t$, then \lcl{} can be solved in time $O(t) = O(1)$, as it suffices to check if we are within distance $\Theta(t)$ from the corner and otherwise we can just output padding.

In our case of toroidal grids, this no longer holds. It is trivial to decide if a given \lcl{} can be solved in $O(1)$ time; only trivial problems in which a constant output is a feasible solution admit an $O(1)$-time solution in toroidal grids. For example, the problem constructed by \citet{naor95what} is now trivial, as there are no corners, and we can always output padding.

We develop a different \lcl{} problem $L_M$ that forces any efficient algorithm to \emph{create corners}. The problem is defined so that the grid can be partitioned in ``tiles'' of arbitrary dimensions, but there are additional requirements:
\begin{enumerate}[noitemsep]
    \item inside each tile we have to solve an inherently global problem, and
    \item in the ``corner'' of each tile we must output the complete execution table of $M$.
\end{enumerate}
Now property~(1) prevents efficient algorithms from producing an output that says that the entire grid consists of one tile. But as soon as the algorithm creates some tile boundaries, property~(2) kicks in and makes sure that we can have finite tiles if and only if $M$ halts in finite time. Additional care is needed to make sure that the problem is solvable but global if $M$ does not halt, as we will discuss next.

\paragraph{\boldmath \lcl{} problem $L_M$ in detail.} For each Turing machine $M$, we define $L_M$ as the disjoint union of two locally checkable labellings $P_1$ and $P_2$; to solve $L_M$, one has to solve either $P_1$ or $P_2$. The problem $P_1$ is defined to 3-colouring in order to make sure that $L_M$ can always be solved in time $O(n)$, independent of $M$. On the other hand, 3-colouring requires time $\Omega(n)$ by \theoremref{thm:3-colouring-lower-bound}. The problem $P_2$ is a problem that involves labelling the grid with the execution table of $M$, started on an empty tape. This problem is formulated so that it can be solved in time $O(\log^* n)$ if and only if $M$ halts on the empty tape.

Each node is labelled with the Turing machine $M$ and a \emph{type}: each node is either an \emph{anchor} or belongs either into one of the four \emph{quadrants} \textsf{NW}, \textsf{NE}, \textsf{SE}, and \textsf{SW}, or one of the four \emph{borders} \textsf{N}, \textsf{S}, \textsf{E}, and \textsf{W}. We overload the notation and define incidence operators as follows. For an arbitrary node $v = (x,y)$, define
\[
\begin{aligned}
  \mathsf{NW}(v) &= (x-1, y+1), & \mathsf{NE}(v) &= (x+1, y+1), \\
  \mathsf{SE}(v) &= (x+1, y-1), & \mathsf{SW}(v) &= (x-1,y-1),  \\
  \mathsf{N}(v)  &= (x, y+1),   & \mathsf{S}(v)  &= (x, y-1),   \\
  \mathsf{E}(v)  &= (x+1, y),   & \mathsf{W}(v)  &= (x-1, y).
\end{aligned}
\]
We will use the types of the nodes to refer to the corresponding incidence operators. The idea of the type labels is that they can be followed to find an anchor.

Let $Q(u) \in \{ \mathsf{NE}, \mathsf{SE}, \mathsf{SW}, \mathsf{NW}, \mathsf{N}, \mathsf{E}, \mathsf{S}, \mathsf{W}, \mathsf{A} \}$, denote the type of a node, and $x(u) \in \{0,1\}$ a colouring. Define the diagonal neighbour $\diag(u)$ of node $u$ as the node reached by taking a step in direction $Q(u)$. For example, if $Q(u) = \mathsf{NW}$, then $\diag(u) = \mathsf{NW}(u)$. For completeness, define the diagonal of an anchor is the node itself.

We have the following local rules.
\begin{enumerate}[noitemsep]
  \item If $Q(u) = \mathsf{NE}$, then $Q(\diag(u)) \in \{ \mathsf{NE}, \mathsf{N}, \mathsf{E}, \mathsf{A} \}$.
  \item If $Q(u) = \mathsf{SE}$, then $Q(\diag(u)) \in \{ \mathsf{SE}, \mathsf{S}, \mathsf{E}, \mathsf{A} \}$.
  \item If $Q(u) = \mathsf{SW}$, then $Q(\diag(u)) \in \{ \mathsf{SW}, \mathsf{S}, \mathsf{W}, \mathsf{A} \}$.
  \item If $Q(u) = \mathsf{NW}$, then $Q(\diag(u)) \in \{ \mathsf{NW}, \mathsf{N}, \mathsf{W}, \mathsf{A} \}$.
\end{enumerate}
On the borders, we must have that $Q(\diag(u)) = Q(u)$, or that $Q(\diag(u)) = \mathsf{A}$. In addition, we require that the borders are surrounded with different labels. In particular we must have that
\begin{enumerate}[noitemsep]
  \item If $Q(u) = \mathsf{N}$, then $Q(\mathsf{W}(u)) = \mathsf{NE}$ and $Q(\mathsf{E}(u)) = \mathsf{NW}$.
  \item If $Q(u) = \mathsf{S}$, then $Q(\mathsf{W}(u)) = \mathsf{SE}$ and $Q(\mathsf{E}(u)) = \mathsf{SW}$.
  \item If $Q(u) = \mathsf{E}$, then $Q(\mathsf{N}(u)) = \mathsf{SE}$ and $Q(\mathsf{S}(u)) = \mathsf{NE}$.
  \item If $Q(u) = \mathsf{W}$, then $Q(\mathsf{N}(u)) = \mathsf{SW}$ and $Q(\mathsf{S}(u)) = \mathsf{NW}$.
\end{enumerate}
Finally for any anchor $v$, we must have that $Q(\mathsf{N}(v)) = \mathsf{S}$, $Q(\mathsf{NW}(v)) = \mathsf{SE}$, $Q(\mathsf{E}(v)) = \mathsf{W}$, $Q(\mathsf{SE}(v)) = \mathsf{NW}$, $Q(\mathsf{S}(v)) = \mathsf{N}$, $Q(\mathsf{SW}(v)) = \mathsf{NE}$, $Q(\mathsf{W}(v)) = \mathsf{E}$, and $Q(\mathsf{NE}(v)) = \mathsf{SW}$.

\begin{figure}
  \centering
  \includegraphics[width=\columnwidth]{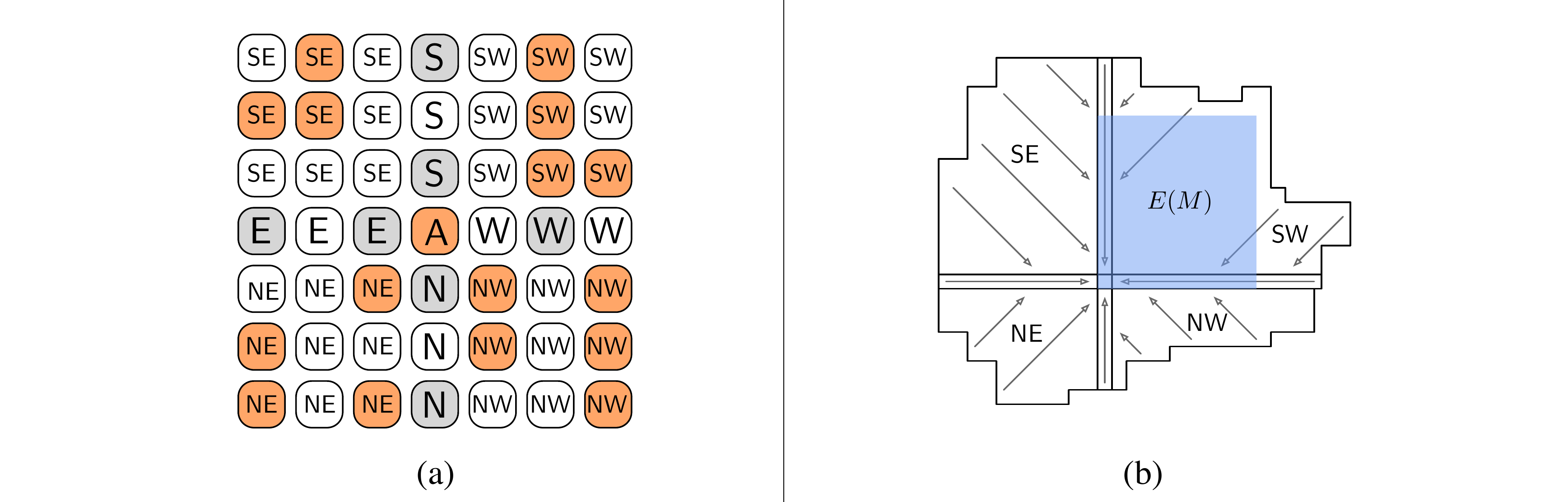}
  \caption{(a) The local rules for the labelling problem $P_2$: all nodes must have a type, indicated by the label, and their diagonal neighbours must have a type that is compatible. Every diagonal must be 2-coloured, that is, for every pair $u,v$ such that $u = \diag(v)$ and $Q(u) = Q(v)$ we must have that $x(u) \neq x(v)$. The anchor must be surrounded by the other labels. (b) The general structure of $O(\log^* n)$ time solution to $P_2$ if $M$ halts. The Voronoi tile of an anchor is split into four quadrants and four borders. Every diagonal can be followed to reach the anchor. An encoding of the execution table $E(M)$ starts from the anchor and is contained inside the Voronoi tile.}
  \label{fig:undecidability}
\end{figure}

In addition, the diagonals must be 2-coloured, that is, we require that if $Q(u) = Q(\diag(u))$, then $x(u) \neq x(\diag(u))$. This condition ensures that any fast solution cannot have large (e.g.\ linear-sized) contiguous fragments of nodes with the same type, and that anchor nodes must appear in the solution.

Finally, we require that starting from each anchor, the grid is labelled with the encoding of the execution table of $M$ when started on an empty tape. This encoding is detailed in the following paragraph.

\paragraph{\boldmath Encoding the execution table of a Turing machine $M$.} Consider an anchor node $v$. We will translate the coordinate system of $G$ so that $v = (0,0)$.

Assume that $M$ runs for $s$ steps on the empty tape. The encoding of the execution table $E(M)$ of $M$ consists of an $(s+1) \times r$ rectangular subgrid, where $r \le s+1$, with the bottom left corner at the anchor $v$. Each row $j$ of $E(M)$ encodes the contents of the tape at the beginning of step $j$. Each column $i$ corresponds to a cell of the tape. Thus, each node $(i,j)$ is labelled with the contents of the cell $i$ before step $j$. In addition, one node on each row holds the machine head and the state of $M$. Each node $u$ on the left boundary of $E(M)$ must have $Q(u) = \mathsf{S}$, and each node $w$ on the bottom boundary must have $Q(w) = \mathsf{W}$.

On the first row each cell is empty and the anchor $v$ holds the machine head. Every $2 \times 2$ subgrid of $E(M)$ must be consistent with the transition rules of $M$. On the last row one of the nodes holds the machine head in a halting state. Each node may hold the encoding of at most one machine. Only nodes with labels $Q \in \{ \mathsf{S}, \mathsf{W}, \mathsf{SW} \}$ may be labelled with an encoding of the Turing machine.

Note that since the labels contain no references to $s$ or the position of any node on $E(M)$, the encoding can be done using a constant number of labels.

\paragraph{Local checkability of the encoding.} Since the nodes can detect if both $P_1$ and $P_2$ are used, we can look at the two cases separately. Clearly a 2-colouring is locally checkable. Now assume that the labelling $P_2$ is used.

The local rules related to the labelling $\ell$ are clearly locally checkable. The nodes can check that they agree on the identity of the Turing machine $M$. The Turing machine encoding is also locally checkable: every anchor and the nodes on the \textsf{W} border can check that the tape is initially empty. Between the rows, nodes can check that the encoding respects the transition rules of $M$. On the top and the right border of the tape nodes can check that the final state is a stopping state and that the encoding is complete.

\paragraph{Solving $L_M$ in time $O(\log^* n)$ if $M$ halts.} Assume that $M$ halts in $s$ steps. $L_M$ can be solved in time $O(\log^* n)$ as follows.
\begin{enumerate}
  \item If $G$ has size $n < 2(s+1)$, solve $P_1$ by brute force.
  \item Else, find a maximal independent set $I$ in $G^{(4(s+1))}$. This is the set of anchors.
  \item Construct the Voronoi tiling $T$ of the anchor set $I$, breaking ties in an arbitrary but consistent manner. The size of each tile is bounded by a constant. Inside each tile $T(v)$ with anchor $v = (x,v)$, label nodes according their position with respect to the anchor:
  \begin{equation} \label{eq:def-quadrants}
  Q(u) =
  \begin{cases}
    \textsf{NW}, \text{ if } x_u > x, y_u < y,\\
    \textsf{NE}, \text{ if } x_u < x, y_u < y,\\
    \textsf{SW}, \text{ if } x_u > x, y_u > y,\\
    \textsf{SE}, \text{ if } x_u < x, y_u > y.
  \end{cases}
  \end{equation}
  Similarly, the borders are labelled as follows.
  \begin{equation} \label{eq:def-borders}
  Q(u) =
  \begin{cases}
    \textsf{N}, \text{ if } x_u = x, y_u < y,\\
    \textsf{S}, \text{ if } x_u = x, y_u > y,\\
    \textsf{E}, \text{ if } x_u < x, y_u = y,\\
    \textsf{W}, \text{ if } x_u > x, y_u = y.
  \end{cases}
  \end{equation}

  From each anchor, start a labelling with the execution table of $M$ as described above. The distance of at least $4(s+1)$ between anchors guarantees that each Voronoi tile can fit the execution table encoding inside it.
\end{enumerate}
Everything is constant time, except finding the maximal independent set, which can be done in time $O(\log^* n)$.

\paragraph{Solving $L_M$ requires time $\Omega(n)$ if $M$ does not halt.} Now assume that $M$ does not halt on the empty tape. Solving $P_1$ naturally requires $\Omega(n)$ time. There are two possibilities for the labelling $P_2$: either the labelling contains an anchor, or not.

First, assume that the labelling contains an anchor $v = (x,y)$. This means that around the anchor, the grid must be labelled with the execution table of $M$, starting with an empty tape. The nodes $(x, y+j)$, with $j > 0$, must be labelled with $\mathsf{S}$ and the contents of the first cell of $M$'s tape before time steps $j$. The nodes $(x+i, y)$, with $i > 0$, must be labelled with $\mathsf{W}$ and the initial, empty contents of $M$'s tape. Since $M$ does not halt on the empty tape, either some node detects and illegal transition in the encoding of the execution table, or the table wraps around the grid. Then there must be a node labelled with $\mathsf{N}$ or $\mathsf{NW}$, and contents of $M$'s tape, a contradiction to the correctness of the output.

Now assume that there are no anchors. If there are no borders, all nodes must be labelled with the same quadrant $Q \in \{ \mathsf{NW}, \mathsf{NE}, \mathsf{SE}, \mathsf{SW} \}$, as otherwise there would a node with the wrong type of diagonal neighbour. Then we can find diagonals of length $\Omega(n)$ that must be 2-coloured, requiring time $\Omega(n)$. Now assume that there exists a node labelled with a border. Since there are no anchors, this node must have a diagonal labelled with the same border, until the border wraps around. The border has length $\Omega(n)$ and must again be 2-coloured, leading to a running time of $\Omega(n)$.

\paragraph{Solving $L_M$ requires time $\Omega(\log^* n)$ if $M$ halts.} Finally, we note that solving $L_M$ requires time $\Omega(\log^*n)$, as it requires breaking symmetry between nodes.

We have shown that $L_M$ has an $O(\log^* n)$ time algorithm if and only if $M$ halts on an empty tape. This is known to be an undecidable problem, and therefore the problem of deciding whether an $O(\log^* n)$ time algorithm exists is in general also undecidable.

\section{Synthesis}\label{sec:synthesis}

At first, the undecidability result of \sectionref{sec:undecidability} seems to suggest that there is little hope in automating algorithm design for \lcl{} problems in grids. Indeed, given an \lcl{} problem $P$, we cannot even tell if it can be solved in $O(\log^* n)$ time or if it is inherently global.

However, in a sense this is the \emph{only} obstacle for automatic synthesis of optimal algorithms! Let us assume that we are given $1$ bit of advice indicating whether $P$ is local (solvable in time $O(\log^* n)$) or global. We will now argue that this information is enough to automatically synthesise an asymptotically optimal algorithm for~$P$.

If $P$ is global, then there is a trivial brute-force algorithm of time $O(n)$ that merely gathers the entire output at a single node and solves the problem globally.

If $P$ is local, we can first check whether it is trivial. If there is a constant label that can be used to fill the entire grid, then (and only then) the problem is solvable in time $O(1)$.

The remaining case is a local problem that cannot be solved in time $O(1)$. Now \theoremref{thm:speedup} and the classical result of \citet{naor95what} imply that the only possibility is the complexity of $\Theta(\log^* n)$. Moreover, the proof of \theoremref{thm:speedup} suggests a convenient \emph{normal form}: problem $P$ can be solved with an algorithm of the form $A' \circ S_k$ for some constant $k$, where
\begin{itemize}[noitemsep]
    \item $S_k$ finds a set of anchors $I$ that forms a maximal independent set in $G^{(k)}$,
    \item $A'$ is an algorithm with running time bounded by $O(k)$ that takes as an input only the set of anchors $I$ and the global orientation of the grid.
\end{itemize}
In the proof of \theoremref{thm:speedup}, algorithm $A'$ first constructs Voronoi tiles, then assigns locally unique identifiers, and then simulates some $O(k)$-time algorithm $A$. But we do not need to worry about such details here; we can see this entire process as a black box $A'$ that simply takes the placement of anchors in the radius-$O(k)$ neighbourhood as input, and using only this information produces the final local output. In particular, $A'$ does not depend on the assignment of unique identifiers or on the value of $n$.

It follows that $A'$ is a \emph{finite function}, mapping radius-$O(k)$ neighbourhoods in a $\{0,1\}$-labelled grid to local outputs. There are only finitely many ways to assign $\{0,1\}$ labels in a constant-sized fragment of the grid, and hence $A'$ can be conveniently represented as a finite lookup table.

The only missing piece is finding the value of $k$, and to do that, we can simply start with $k = 1$ and increment it until synthesis succeeds. (Note that if we were dealing with a global problem instead of a local problem this loop will never terminate.)

For each value of $k$, we proceed as follows. We pick sufficiently large values $r_1, r_2 = \Theta(k)$. Then we enumerate all possible ways in which the anchors may appear within a $r_1 \times r_2$ fragment of the grid; these are called \emph{tiles}. We describe in Appendix~\ref{app:tiles} a practical algorithm for such an enumeration. For example, for $k = 1$ we have the following $3 \times 2$ tiles; if we consider a maximal independent set in the grid, and pick a $3 \times 2$ window, we will see one of these configurations:
\newcommand{\T}[3]{\,\boxed{\arraycolsep=0pt\renewcommand{\arraystretch}{0}\begin{array}{c}#1\\[4pt]#2\\[4pt]#3\end{array}}\,}
\[
\T{00}{00}{10}\T{00}{00}{01}\T{00}{10}{00}\T{00}{10}{01}\T{00}{01}{00}\T{00}{01}{10}\T{10}{00}{00}\T{10}{00}{10}\T{10}{00}{01}\T{10}{01}{00}\T{10}{01}{10}\T{01}{00}{00}\T{01}{00}{10}\T{01}{00}{01}\T{01}{10}{00}\T{01}{10}{01}
\]
We will then construct a \emph{neighbourhood graph} $H = (V_H,E_H)$, in which each node $u \in V_H$ corresponds to a $r_1 \times r_2$ tile, and each edge corresponds to a tile of dimensions $(r_1 + 1) \times r_2$ or $r_1 \times (r_2 + 1)$. For example, there is a $3 \times 3$ tile
\[
\T{000}{010}{100}
\]
and hence in the neighbourhood graph of $3 \times 2$ tiles there is a directed \emph{horizontal edge}
\[
\Biggl(\T{00}{01}{10},\T{00}{10}{00}\Biggr).
\]
Similarly, we can identify directed vertical edges. Now $A'$ is simply a mapping from $V_H$ to local outputs; $A'(u)$ is what we output for a node whose local neighbourhood with respect to $I$ is equal to $u$. Furthermore, the constraints of the \lcl{} problem $P$ (once sufficiently normalised) can be encoded as constraints related to horizontal and vertical edges. For example, in the $4$-colouring problem, the constraint is simply that adjacent tiles have different labels.

Hence the task of synthesising algorithm $A'$ reduces to a combinatorial constraint satisfaction problem in which our task is to find a labelling of the nodes of graph $H$ that satisfies all constraints on the edges of the graph; if such an assignment does not exist, we simply repeat the process with a larger value of $k$ and larger tile dimensions.

We have successfully used this approach with many concrete \lcl{} problems discussed in this work, and it works well in practice. As a concrete nontrivial example, consider the problem of $4$-colouring 2-dimensional grids. Here it can be shown that no solution exists for $k = 1$ or $k = 2$, but synthesis succeeds with $k = 3$ for e.g.\ $7 \times 5$ tiles. While a priori it might seem that the number of tiles is impractical for such parameter values ($2^{7 \cdot 5}$ candidate tiles?), the key observation is that $1$'s are fairly sparse in any maximal independent set of $G^{(k)}$, and it turns out that we only need to consider $2079$ tiles. Finding a proper $4$-colouring of the neighbourhood graph can be done with modern SAT solvers in a matter of seconds.

\section{\boldmath Vertex colouring \texorpdfstring{$d$}{d}-dimensional grids with 4 colours}\label{app:vertex-4-colouring}

In this section we prove the following general upper bound. We remark that this result can be almost directly derived also from the work of Holroyd et al.~\cite{finitarycol} (see Corollary~15 in particular), but we give a direct proof here for the sake of completeness.
\begin{theorem}
\label{thm:vertex-colouring}
For every fixed $d \ge 2$, the complexity of $4$-colouring $d$-dimensional grids is $\Theta(\log^{*} n)$.
\end{theorem}

Here we consider $d$-dimensional (toroidal) grids, for some fixed dimension $d$. Particularly, each of the $n^d$ vertices $v$ has $d$ coordinates: $v = (v_1,v_2,\ldots,v_d)$, where $v_i \in \{0,1,\ldots,n-1\} = [n]$. For the sake of simplifying notation, we will not distinguish between a vertex and the vector of its coordinates, and treat all arithmetic operations on coordinates as happening in a $(\operatorname{mod} n)$ regime, thus for any $u,v \in [n]^d$ we have $uv,u-v \in [n]^d$. For $x \in [n]$ we define $\Vert x \Vert = \min\{x,n-x\}$,  and for $v \in [n]^d$  the $L_1$ norm as $\Vert v \Vert = \sum_{1 \le i \le d} \Vert v_i \Vert$ and the $L_\infty$ norm as $\Vert v \Vert_{\infty} = \max_{1 \le i \le d} \Vert v_i \Vert$.
The $L_1$ and $L_\infty$ distance definition follows from the corresponding norm definition.
Observe that $L_1$ distance corresponds to the distance using grid edges.

\begin{definition}
	We define the \emph{radius-$r$ ball} of $u$ as
\[B_\infty(u,r) = \{v : \Vert u-v\Vert_\infty \le r\}.\]
	Moreover, we denote by $G^{[k]}$ the $k$th power of $G$ according to the $L_{\infty}$ norm, i.e., $V(G^{[k]}) = V(G)$ and
\[
	E(G^{[k]}) = \bigl\{ \{u,v\} : \Vert u - v \Vert_{\infty} \le k \bigr\}.
\]
\end{definition}

We also need the notion of \emph{conflict colouring}, as given by \citet{fraigniaud16local}.
\begin{definition}
Given graph $G$, lists of available colours for each vertex and lists of forbidden colour pairs for each edge, we say that a problem of assigning colours to vertices so that: (i) each vertex is assigned one of colours from this list and (ii) no edge observes on its endpoints a pair of colours from forbidden pair; is a $(\ell,d)$-conflict colouring problem, if:
\begin{enumerate}
\item each available colour list is of length at least $\ell$,
\item for each edge, for each colour on one endpoint, there are at most $d$ forbidden colours on the other endpoint.
\end{enumerate}
\end{definition}
\citet{fraigniaud16local} show that if $\ell / d > \Delta$, then there is a distributed algorithm solving $(\ell,d)$-conflict colouring in $\widetilde{O}(\sqrt{\Delta})+\log^*n$ rounds. However, we observe that a greedy approach gives a good enough running time for our purposes: (i) colour vertices of graphs using $\Delta^2$ colours (classical vertex-colouring problem) (ii) in $\Delta^2$ rounds, iterate through colours, in round $i$ vertices of colour $i$ take any colour in a greedy fashion.

\begin{proof}[Proof of Theorem~\ref{thm:vertex-colouring}]
Let us name a parameter $\ell$, of even value to be fixed later. We use a set of anchors $M$ being the maximal independent set of vertices of $G^{[\ell]}$. Since the degree of a vertex in $G^{[\ell]}$ is at most $(2\ell+1)^d$, $M$ can be found in $O((2\ell+1)^{2d}+\log^{*}n)$ rounds on $G^{[\ell]}$. Since $\Vert\cdot\Vert_{1} \le d  \Vert\cdot\Vert_{\infty}$, any algorithm can be simulated on $G$ with $\ell \cdot d$ multiplicative slowdown, giving in total $O(\ell \cdot d \cdot (2\ell+1)^{2d}+ \ell \cdot d \cdot \log^{*}n))$ rounds.

Our aim is to assign to every vertex of $v \in M$ a radius $r(v) \in \mathbb{Z}^{+}$, such that:
\begin{enumerate}
\item \label{properties_enum_1} $\{B_\infty(v, r(v)-1) : v \in M\}$ covers all $V$,
\item \label{properties_enum_2} for any $u,v \in M$ such that if $B_\infty(u,r(u)+1) \cap B_\infty(v,r(v)+1) \not= \emptyset$ then the bounding hyperplanes for those $L_\infty$ balls are separated, that is,
\[\forall_{1 \le i \le d} \min_{\varepsilon_1,\varepsilon_2 \in \{-1,1\}} \bigl\Vert (u_i + \varepsilon_1 \cdot r(u)) - (v_i + \varepsilon_2 \cdot r(v)) \bigr\Vert \ge 2. \]
\end{enumerate}

Consider the family of  $L_\infty$ balls of radius $\ell$ centred in every vertex of $M$: $\{ B_\infty(v,\ell) : v \in M \}$.
By the properties of MIS, this family covers every vertex of $V$, as otherwise we could add one more vertex to $M$, thus to satisfy \ref{properties_enum_1} it is enough to have $r(v) > \ell$.

Next we show that for large enough $\ell$ we can find an appropriate assignment of radii fast, by reduction to local conflict colouring, with colours $\ell < r(v) < 2\ell$.

\begin{figure}
\begin{minipage}{.5\textwidth}
\centering\includegraphics[scale=0.4]{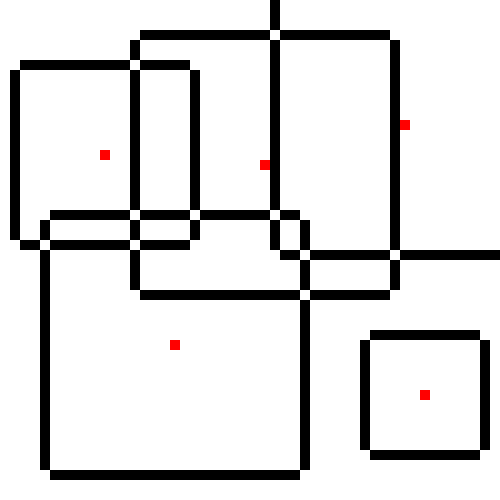}
\end{minipage}
\begin{minipage}{.5\textwidth}
\centering\includegraphics[scale=0.4]{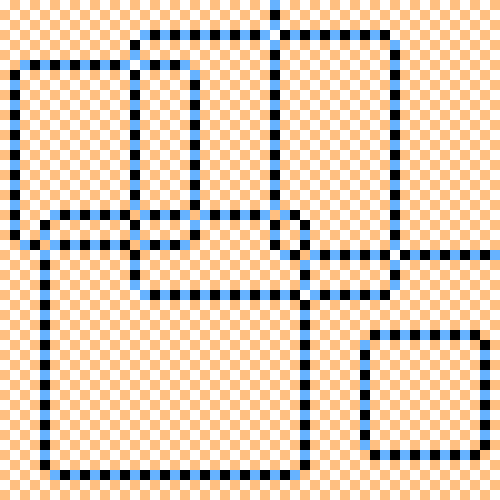}
\end{minipage}
\caption{On the left, partition resulting from a choice of ball centres and ball radii. Induced vertex 4-colouring on the right.}
\label{fig:4coloring_example}
\end{figure}

\begin{lemma}
\label{lem:balls}
Consider family of $L_\infty$ balls of radius $c \ell$ centred in every vertex of $M$. Then every such ball intersects with at most $(8c)^d$ other balls.
\end{lemma}
\begin{proof}
Consider $u,v \in M$, $u \not= v$.
By the triangle inequality $\Vert u-v \Vert_\infty > \ell$, thus $B_\infty(v,\ell/2) \cap B_\infty(v,\ell/2) = \emptyset$.
Also, $B_\infty(v,c\ell)$ intersects $B_\infty(u,c\ell)$ iff $v \in B_\infty(u,2c\ell)$, but that implies that $\left| B_\infty(v, \ell/2) \cap B_\infty(u,2c \ell)\right| \ge (\ell/2+1)^d$.
Using bounds on first, the fact that all $B_\infty(v,\ell/2), v \in M$ balls are disjoint, and then that if the points are centres of balls of radius $c\ell$ that are intersecting with ball centred in $u$, then they intersect on large volume, we can bound the total number of intersecting balls as
\[ \frac{ |B_\infty(u,2c \ell)| }{ (\ell/2+1)^d } = \left(\frac{4 c\ell + 1}{\ell/2+1}\right)^d \le (8c)^d.\qedhere\]
\end{proof}

Instead of considering conflicts over intersection of balls $B_\infty(v,r(v))$, we will guarantee no conflicts over intersections of $B_\infty(v,2\ell)$.
Now consider the graph $H$, over vertex set $M$, with edges connecting every pair of $u,v$ such that $B_\infty(u,r(u)+1) \cap B_\infty(v,r(v)+1) \not= \emptyset$.
By Lemma~\ref{lem:balls}, the maximum degree in $H$ can be upper-bounded as $\Delta_H \le 16^d$.

While vertices do not know their coordinates, that is $v$ does not have information on the values of $(v_1,\ldots,v_n)$, a pair of vertices $u,v$ such that $(u,v) \in H$ is able to determine $u_i-v_i$ for all $i$. To satisfy condition \ref{properties_enum_2} it is enough to exclude at most $12$ possible values of $r(u)$, per each dimension and each value of $r(v)$. That is, we need to ensure that
\begin{equation}
\label{eq:conflict_coloring}
\forall_{(u,v) \in H} \forall_{1 \le i \le d} \forall_{\varepsilon_1,\varepsilon_2 \in \{-1,1\}} \forall_{x \in \{-1,0,1\}}\ \  \varepsilon_1 \cdot r(u)  \not= x  + \varepsilon_2\cdot r(v) + (v_i - u_i).
\end{equation}
Thus our problem is an $(\ell,12d)$-conflict colouring, and can be solved locally if $(\ell-1) / (12d) > \Delta_H$, so it is enough to set $\ell = 1+12 d \cdot 16^d$ for our purposes. Running time is upper-bounded by $O(\text{poly}(\Delta_H) + \log^{*} n)$ rounds in $H$, which can be simulated with multiplicative overhead of $\ell$, giving total time of this part $O(\ell \cdot 16^{O(d)} + \ell \cdot \log^{*}(n))$ rounds.

What remains to show, is that given $M$ and all $r(v)$, we can compute locally a $(2,2 d  \ell)$ weak diameter network decomposition of $G$: a decomposition $V = V_1 \cup V_2$ into disjoint sets, such that each connected component of $V_i$ is of diameter $2 d \ell$. Given such a network decomposition, a $4$ colouring of $G$ can be found in $O(d \ell)$ rounds trivially.

We say that $u \in V$ is on the $i$th border of $v \in M$, if $u_i \in \{v_i-r(v),v_i+r(v)\}$ and $\Vert u - v \Vert - r(v)$.
By the property \ref{properties_enum_2} of radii, we know that every vertex $v$ is on the $i$th dimension border of at most one vertex $u \in M$, and it can be decided locally during the computation of the radii. We define
\[\textrm{count}(v) = \bigl|\bigl\{ (i,u) : v \text{ is on the }i\text{-th dimension border of }u\bigr\}\bigr|.\]
We assign $v$ to $V_1$ iff $\textrm{count}(v)$ is odd and to $V_2$ iff $\textrm{count}(v)$ is even.

\begin{lemma}
\label{lem:bordercrossing}
If $u,u' \in V$ are neighbouring in $G$, and there is $v$ such that: $\Vert u - v \Vert_\infty = r(v)-1$ and $\Vert u' - v \Vert_\infty = r(v)$, then $\textrm{count}(u) + 1 = \textrm{count}(u')$.
\end{lemma}

\begin{proof}
Let $j$ be the dimension such that $u_j \not= u'_j$. Then  $\Vert u'_j - v_j \Vert = r(v)$, $\Vert u_j - v_j \Vert = r(v)-1$ and $\forall_{i \not= j} \Vert u'_i - v_i \Vert = \Vert u_i - v_i \Vert \le r(v)-1$, that is $j$ was the only dimension on which it was on the border of $v$. Moreover, $u$ cannot be on the $j$th dimension border for any other vertex $v'$, as then $v$ and $v'$ would violate property \ref{properties_enum_2} of radii over the $j$th coordinate.

Now we observe that, while $u'$ might be on the $i$th dimension border for some $w$, $i \not= j$, $w \not= v$, those remain the same for $u$. Namely, if we assume otherwise, that is that $\Vert u' - w \Vert_\infty = r(w)$ and $\Vert u - w \Vert_\infty \not= r(w)$, then by simple observation that only the $j$th coordinate changes in those difference vectors, we would have that $u'$ is on $j$th dimension border for $w$, a contradiction.
By analogous reasoning, we have that for any $i$th dimension border that $u$ is on, it remains the same for $u'$.

All in all, we have that $u$ is on one less dimension border than $u'$.
\end{proof}

Now we proceed to show that every connected component of $V_1$ or $V_2$ is a subset  fully contained in $B_\infty(v,r(v)-1)$ for some $v \in M$. Let us assume that this is not the case. Take any connected component $X$ and $u \in X$, and let $v$ be such that $u \in B_\infty(v,r(v)-1)$ (by property \ref{properties_enum_1} there is always one). If $X \not\subseteq B_\infty(v,r(v)-1)$, then there are neighbouring $u',u'' \in X$, such that $\Vert u' - v \Vert_\infty = r(v)-1$ and $\Vert u''- v \Vert_\infty =r(v)$. However, by Lemma~\ref{lem:bordercrossing} they cannot be on the same side of the partition, a contradiction.
\end{proof}

\section{Lower bound for 3-colouring 2-dimensional grids}\label{app:vertex-3-colouring}

\begin{theorem}\label{thm:3-colouring-lower-bound}
The complexity of $3$-colouring on $2$-dimensional grids is $\Omega(n)$.
\end{theorem}

The rough outline of the proof is as follows:
\begin{itemize}[noitemsep]
    \item We first show that a certain artificial coordination problem requires $\Omega(n)$ rounds on directed cycles.
    \item We then reduce this problem to $3$-colouring two-dimensional grids. Essentially, we show that any $3$-colouring algorithm for grids solves an instance of the aforementioned coordination problem for each row of the grid.
\end{itemize}
As with the $4$-colouring upper bound, the general idea of the proof is very similar to the one used by Holroyd et al.~\cite{finitarycol}. However, directly translating the proof seems more difficult in this case due to subtle differences between the models.

\paragraph{\boldmath The $q$-sum coordination problem.} Let $q \colon \N \to \Z$ be a function. In the \emph{$q$-sum coordination problem}, we assume that the input graph is a directed cycle with unique identifiers, and each node $v$ has to output $\ell(v) \in \{ -1, 0, 1 \}$ such that $\sum_{v \in V} \ell(v) = q(n)$, where $n = |V|$. That is, this is a family of problems, one for each possible function $q$. We now show that this problem is global for even fairly simple choices of $q$.

\begin{theorem}\label{thm:q-sum-lower-bound}
Let $q \colon \N \to \Z$ be a function such that
\begin{enumerate}[noitemsep]
    \item $q(n)$ is odd when $n$ is odd, and
    \item $|q(n)| \le n/2$ for all $n$.
\end{enumerate}
Then $q$-sum coordination requires $\Omega(n)$ rounds.
\end{theorem}

\begin{proof}

Assume that we have an algorithm $A$ that solves the problem in $T(n)=o(n)$ rounds. Fix a sufficiently large odd $n$ such that $T(n) < n/200$. We show that we can construct an identifier assignment for a directed cycle of length $n$ for which the sum of the outputs of $A$ is greater than $n/2$, giving a contradiction.

We say that an \emph{input fragment} $F$ is a sequence of unique identifiers. We may interpret an input fragment $F$ as a connected subgraph of a possible input graph of size $n$; we denote the length of sequence $F$ by $|F|$, and say that fragments $F_1$ and $F_2$ are disjoint if the corresponding identifier sets are disjoint. Given at least two disjoint input fragments $F_1, F_2, \dotsc, F_k$ and $|F_i| \ge n/100$, we define $A(F_1 F_2 \dotsb F_k)$ as the sum of output labels $A$ gives to vertices from the midpoint of $F_1$ (inclusive) to the midpoint of $F_k$ (exclusive) in the subgraph corresponding to the concatenated sequence $F_1 F_2 \dotsb F_k$. Note that since $T(n) < n/100$, this value only depends on $F_1, F_2, \dotsc, F_k$. Moreover, denote by $P(F_1 F_2 \dotsb F_k)$ the parity of $A(F_1 F_2 \dotsb F_k)$. It follows immediately from the definition that $P(F_1 \dotsb F_j \dotsb F_k) = P(F_1 \dotsb F_j) +  P(F_{j} \dotsb F_k)$.

\begin{lemma}\label{lemma:A-change-value}
There are disjoint input fragments $F_1$ and $F_2$ with $|F_1| = |F_2| = \lceil n/100 \rceil$ such that for some input fragments $X_1$, $X_2$ disjoint from $F_1$ and $F_2$ with $|X_1|, |X_2| \in [2n/100,96n/100]$ we have $P(F_1 X_1 F_2) \not= P(F_1 X_2 F_2)$.
\end{lemma}

\begin{proof}
Assume that the claim does not hold. Then, for any disjoint input fragments $F_1$ and $F_2$ with $|F_1| = |F_2| = \lceil n/100 \rceil$ there is a value $P(F_1 \mathord{*} F_2)$ such that $P(F_1 X F_2) = P(F_1 \mathord{*} F_2)$ for all $X$ with $|X| \in [2n/100,96n/100]$. By considering a cycle of form $F_1 X_1 F_2 X_2$, where all fragments are disjoint, $|X_1|, |X_2| \in [2n/100,96n/100]$ and $|F_1|+|X_1|+|F_2|+|X_2|=n$, we observe that for any $F_1$ and $F_2$ we have that $P(F_1 \mathord{*} F_2)$ and $P(F_2 \mathord{*} F_1)$ have fixed, different values, since
\[ P(F_1 \mathord{*} F_2) + P(F_2 \mathord{*} F_1) = P(F_1 X_1 F_2) + P(F_2 X_2 F_1) = q(n)\,,\]
which is odd. Moreover, fixing disjoint $F_1$, $F_2$ and $F_3$ such that $P(F_1 \mathord{*} F_2) = 0$ and considering a length-$n$ cycle of form $F_1 X_1 F_2 X_2 F_3 X_3$, we observe by a similar argument that either $P(F_2 \mathord{*} F_3) = 0$ or $P(F_3 \mathord{*} F_1) = 0$; by relabelling $F_1$, $F_2$ and $F_3$ if necessary, we can assume that $P(F_1 \mathord{*} F_2) = P(F_2 \mathord{*} F_3) = 0$. Considering disjoint fragments $Y_1$ and $Y_2$ with $|Y_1| = |Y_2| = \lceil 2n/100 \rceil$, we finally observe that
\[ P(F_1 \mathord{*} F_3) = P(F_1 Y_1 F_2 Y_2 F_3) = P(F_1 Y_1 F_2) + P(F_2 Y_2 F_3) = P(F_1 \mathord{*} F_2) + P(F_2 \mathord{*} F_3) = 0 \,.\]

Now assume that $F_1$, $F_2$ and $F_3$ are disjoint fragments with $|F_1| = |F_2| = |F_3| = \lceil n/100 \rceil$ as above, and let $X$ be a fragment disjoint from $F_1$, $F_2$ and $F_3$ with $|X| =  \lceil 2n/100 \rceil$. We now have the following:
\begin{align}
    P(F_1 X F_2) = 0 \hspace{5mm} \Rightarrow \hspace{5mm} & P(F_1 X) = P(X F_2)\,,\label{eq:cyc1}\\
    P(F_2 X F_3) = 0 \hspace{5mm} \Rightarrow \hspace{5mm} & P(F_2 X) = P(X F_3)\,,\label{eq:cyc2}\\
    P(F_1 X F_3) = 0 \hspace{5mm} \Rightarrow \hspace{5mm} & P(F_1 X) = P(X F_3)\,.\label{eq:cyc3}
\end{align}
Thus, we have
\begin{equation}
P(F_2 X) \overset{\text{(\ref{eq:cyc2})}}{=} P(XF_3) \overset{\text{(\ref{eq:cyc3})}}{=} P(F_1 X) \overset{\text{(\ref{eq:cyc1})}}{=} P(XF_2)\,. \label{eq:cyc4}
\end{equation}
Furthermore, we have that
\begin{align}
    P(F_3 X F_2) = 1 \hspace{5mm} \Rightarrow \hspace{5mm} & P(F_3 X) = P(X F_2) + 1\,, \label{eq:cyc5}\\
    P(F_2 X F_1) = 1 \hspace{5mm} \Rightarrow \hspace{5mm} & P(X F_1) = P(F_2 X) + 1\,. \label{eq:cyc6}
\end{align}
Thus,
\begin{equation*}
1 = P(F_3 X F_1) = P(F_3X) + P(X F_1) \overset{\text{(\ref{eq:cyc5},\ref{eq:cyc6})}}{=} P(X F_2) + 1 + P(F_2 X) + 1 \overset{\text{(\ref{eq:cyc4})}}{=} 0\,,
\end{equation*}
which is a contradiction.
\end{proof}

Now let $F_1$, $F_2$, $X_1$ and $X_2$ be as in \lemmaref{lemma:A-change-value}. If $|X_1| = |X_2|$, we are done, since for any fragment $Y$ disjoint from the other fragments such that $|F_1| + |F_2| + |X_1| + |Y| = n$, the cycles $F_1 X_1 F_2 Y$ and $F_1 X_2 F_2 Y$ are valid instances of $q$-sum coordination with different outputs. Otherwise, we can assume that $|X_1|+1=|X_2|$ without loss of generality; in fact, we can assume that $X_2$ is obtained from $X_1$ by adding an unique identifier to the start of $X_1$, since by the above observation $P(F_1 X F_2)$ is defined by the length of $X$.

Now let $d=A(F_1 X_2 F_2)-A(F_1 X_2 F_2)$, and observe that $|d| \geq 1$. Let $X_3$ be obtained by adding another unique identifier to the start of $X_2$. Consider removing the last identifier from $X_1$, $X_2$ and $X_3$ to obtain $X_1^-$, $X_2^-$ and $X_3^-$, respectively. First, consider $X_2^-$; clearly $A(F_1 X_2^- F_2)=A(F_1 X_1 F_2)$ since $|X_2^-|=|X_1|$. Thus, removing the last identifier $v$ from the end of $X_2$ reduces the sum of the outputs in the $T(n)$-neighbourhood of $v$ by $d$, and since $X_2$ is sufficiently long, this does not effect the first vertices of $X_2$. However, since the local changes look the same within a $T(n)$-radius neighbourhood, this implies that we also have $A(F_1 X_3 F_2)-A(F_1 X_3^- F_2)=d$ and $A(F_1 X_1 F_2)-A(F_1 X_1^- F_2)=d$. That is, adding an identifier to the front of $X_2$ increases the score by $d$, and removing an identifier from the end of $X_1$ decreases the score by $d$.

By repeatedly applying this argument to both directions, we can construct a sequence of fragments $Y_1, Y_2, \dotsc, Y_{\lceil 9n/10 \rceil}$ such that $|Y_1| = \lceil 2n/100 \rceil$, $|Y_{k+1}| = |Y_{k}| + 1$ and $A(F_1 Y_{k+1} F_2) = A(F_1 Y_1 F_2) + kd$. By definition of the problem, $|A(F_1 Y_1 F_2)| \le 4n/100$, so $|A(F_1 Y_{\lceil 9n/10 \rceil} F_2)| \ge 8n/10$ and $|Y_{\lceil 9n/10 \rceil}| \ge 92n/100$. But this means that the sum of outputs of $A$ on any input containing the fragment $F_1 Y_{\lceil 9n/10 \rceil} F_2$ has absolute value more than $n/2$, which is a contradiction.
\end{proof}

\paragraph{\boldmath Reduction to $3$-colouring.} Fix an algorithm $A$ for $3$-colouring grids, and assume $A$ runs in $T(n) = o(n)$ rounds. By adding a constant-round preprocessing step, we may assume that $A$ produces a greedy colouring, that is, if node $v$ has colour $2$, then it has a neighbour of colour $1$, and if it has colour $3$, then it has neighbours of colours $1$ and $2$. We now show that algorithm $A$ can be used to solve $q$-sum coordination in $T(n)$ round for some $q$ satisfying the conditions in \sectionref{thm:q-sum-lower-bound}.

\begin{figure}[b]
\centering
\includegraphics[width=\columnwidth,page=4]{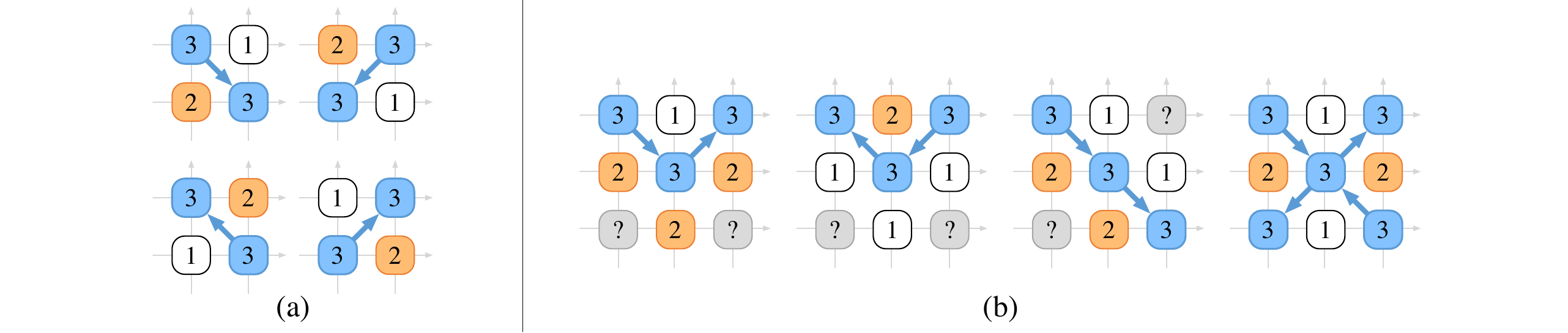}
\caption{(a) Edge directions in $H$. (b) Possible neighbourhoods in $H$ up to rotation.}\label{fig:3-col-cycles}
\end{figure}

Fix the input size $n$ and an input grid $G$, and consider the colouring $c \colon V(G) \to \{ 1,2,3 \}$ produced by $A$. We will now define an auxiliary directed graph $H$ with node set $V(H) = \{ v \in V(G) \colon c(v) = 3\}$ as follows. We add a directed edge to $E(H)$ between two nodes $u,v$ with $c(v) = c(u) = 3$ if they share two neighbours $w,w'$ such that $c(w) = 1$ and $c(w') = 2$, and we direct this edge so that the common neighbour with colour $1$ is to the ``left'' of the edge (\figureref{fig:3-col-cycles}a shows all possibilities). There are four possible neighbourhoods for a node in $H$, up to rotation (\figureref{fig:3-col-cycles}b):
\begin{enumerate}
    \item If $v$ has exactly one neighbour of colour 1 (say to the north), then there is an in-edge from the node to the north-west and an out-edge to the node to the north-east.
    \item If $v$ has exactly one neighbour of colour 2 (again to the north), then there is an in-edge from the node to the north-east and an out-edge to the node to the north-west.
    \item If $v$ has two neighbours of colour 1 (say to the north and east), then there is an in-edge from the node to the north-west and an out-edge to the node to the south-east.
    \item If $v$ has two neighbours of colour 1 (say to the north and south), then there are in-edges from the nodes to the north-west and to the south-east, as well out-edges to the nodes in the north-east and in to the south-west.
\end{enumerate}
In particular, each node has either in-degree $1$ and out-degree $1$, or in-degree $2$ and out-degree $2$ in~$H$.
Thus, we can partition $E(H)$ into a collection $\mathcal{C}$ of edge-disjoint directed cycles.

Consider a cycle $C \in \mathcal{C}$ and a row $r$ of $G$, and let $u, v, w$ be nodes on $C$ such that $(u,v) \in C$ and $(v,w) \in C$. We say that $v$ is a \emph{northbound intersection} if $u$ is on the row south of $v$ and $w$ is on the row north of $v$. Similarly, we say that $v$ is a \emph{southbound intersection} if $u$ is on the row north of $v$ and $w$ is on the row south of $v$. Let $\operatorname{north}_r(C)$ be the number of northbound intersections on $C$ and $\operatorname{south}_r(C)$ the number of southbound intersections on $C$ and define $i_r(C)=\operatorname{north}_r(C)-\operatorname{south}_r(C)$.

\begin{lemma}\label{cla:row}
For all rows $r_1$ and $r_2$, we have that $i_{r_1}(C)=i_{r_2}(C)$.
\end{lemma}

\begin{proof}
It is enough to show that this is the case for two adjacent rows, so let $r_1 = r$ and $r_2 = r+1$, that is, $r_2$ is the row immediately to the north of $r_1$. In the case that $C$ does not intersect either row the claim holds. Otherwise, the set $I = \{ v \in V(C) \colon v \text{ is an intersection on $r_1$ or $r_2$} \}$ is non-empty. For $u,v \in I$, we say $u$ \emph{follows} $v$ if $u$ is the next element of $I$ we reach when following the cycle from $v$ in the direction of the edges; likewise, we say that $v$ \emph{precedes} $u$.

The set $I$ may contains four types of intersections: northbound on $r_1$ (denoted by $\mathsf{N}_{r_1}$), southbound on $r$ ($\mathsf{S}_{r_1}$), northbound on $r_2$ ($\mathsf{N}_{r_2}$) and southbound on $r_2$ ($\mathsf{S}_{r_2}$). We now observe that following hold:
\begin{enumerate}[noitemsep]
    \item $\mathsf{N}_{r_1}$ is followed by $\mathsf{N}_{r_2}$ or $\mathsf{S}_{r_1}$.
    \item $\mathsf{S}_{r_2}$ is followed by $\mathsf{S}_{r_1}$ or $\mathsf{N}_{r_2}$.
    \item $\mathsf{N}_{r_2}$ is preceded by $\mathsf{N}_{r_1}$ or $\mathsf{S}_{r_2}$.
    \item $\mathsf{S}_{r_1}$ is preceded by $\mathsf{S}_{r_2}$ or $\mathsf{N}_{r_1}$.
\end{enumerate}
We prove (1); the other cases follow by a similar argument. Consider an $\mathsf{N}_{r_1}$ intersection $v$. Following the cycle $C$ forward from $v$, we observe that every other node is on row $r_2 = r+1$ and every other node is on row $r_1 = r$, until we either have a node on row $r+2$ or on row $r-1$. That is, the next intersection we encounter is either a northbound intersection on $r_2$ or a southbound intersection on~$r_1$.

For each intersection $v$ of type $\mathsf{N}_{r_1}$ or $\mathsf{S}_{r_2}$, we define the \emph{pair} $p(v)$ of $v$ to be the following intersection on $C$, and for each intersection of type $\mathsf{N}_{r_2}$ or $\mathsf{S}_{r_1}$ we define the pair as the preceding intersection on $C$. By the above case analysis, this partitions $I$ to disjoint pairs $\{ v, p(v) \}$.

Now we observe that there are four possible types of pairs, each of which contributes the same amount to $i_{r_1}(C)$ and $i_{r_2}(C)$:
\begin{enumerate}[noitemsep]
    \item $\mathsf{N}_{r_1}$ and $\mathsf{N}_{r_2}$: contributes $1$ to $i_{r_1}(C)$ and $i_{r_2}(C)$.
    \item $\mathsf{S}_{r_1}$ and $\mathsf{S}_{r_2}$: contributes $-1$ to $i_{r_1}(C)$ and $i_{r_2}(C)$.
    \item $\mathsf{N}_{r_1}$ and $\mathsf{S}_{r_1}$: contributions to $i_{r_1}(C)$ cancel out, contributes nothing to $i_{r_2}(C)$.
    \item $\mathsf{N}_{r_2}$ and $\mathsf{S}_{r_2}$: contributions to $i_{r_2}(C)$ cancel out, contributes nothing to $i_{r_1}(C)$.
\end{enumerate}
Summing over all pairs, we have that $i_{r_1}(C)=i_{r_2}(C)$.

\end{proof}

As a corollary we have that \[\sum_{C \in \mathcal{C}} i_{r_1}(C) = \sum_{C \in \mathcal{C}} i_{r_2}(C)\]
for all rows $r_1$ and $r_2$ in $G$. Writing $s(G)$ for this sum, we make the following claim.

\begin{lemma}
\label{lem:sum_is_const}
We have $s(G_1) = s(G_2)$ for any $n \times n$ grids $G_1$ and $G_2$ when $T(n) < n/4$.
\end{lemma}

\begin{proof}
Construct an $n \times n$ grid $H_1$ from $G_1$ by replacing the unique identifiers on rows $1$ to $\lceil n/2 \rceil$ by identifiers that do not appear in either $G_1$ and $G_2$. Since $T(n) < n/4$, the output on row $\lceil 3n/4 \rceil$ is the same on $G_1$ and $H_1$, so by previous results we have $s(G_1) = s(H_1)$. Constructing a graph $H_2$ from $G_2$ by replacing the identifiers on rows $1$ to $\lceil n/2 \rceil$ with the same ones that appear in $H_1$, we have the same argument that $s(G_2) = s(H_2)$ and, using row $\lceil n/4 \rceil$, that $s(H_1) = s(H_2)$.
\end{proof}

Since $T(n) = o(n)$, there is a constant $n_0$ such that $T(n) < n/4$ for all $n \ge n_0$. Let us define a function $s(n)$ so that if $n < n_0$, then $s(n) = 1$ if $n$ is odd and $s(n) = 0$ if $n$ is even; if $n \ge n_0$, then $s(n) = s(G)$ for any $n \times n$ grid $G$.

\begin{lemma}\label{lemma:cycle-function}
    If $n$ is odd, then $s(n)$ is odd. Moreover, for all $n$, we have $|s(n)| \le n/2$.
\end{lemma}

\begin{proof}
The claim is trivially true when $n < n_0$, so assume $n \ge n_0$. Fix an arbitrary $n \times n$ grid $G$, a row $r$ of $G$ and a $3$-colouring $c \colon V(G) \to \{ 1,2,3 \}$ given by algorithm $A$. Moreover, let $H$ be the auxiliary graph on colour $3$ nodes as before. Assign a label $\ell(v) \in \{ -1, 0, 1 \}$ to each node $v$ on row~$r$:
\begin{enumerate}[noitemsep]
    \item If $c(v) = 3$ and $\indeg_H(v) = \outdeg_H(v) = 1$, let $u,w$ be the unique colour $3$ nodes such that $(u,v) \in E(H)$ and $(v,w) \in E(H)$. We define the label $\ell(v)$ based on the positions of $u$ and $w$ as follows:
    \begin{enumerate}[noitemsep]
        \item $\ell(v) = 1$ if $u$ is on row $r-1$ and $w$ is on row $r+1$,
        \item $\ell(v) = -1$ if $u$ is on row $r+1$ and $w$ is on row $r-1$, and
        \item $\ell(v) = 0$ if $u$ and $w$ are on the same row.
    \end{enumerate}
    \item We define $\ell(v) = 0$ in all other cases.
\end{enumerate}
Informally, this means that $\ell(v) = 1$ if $v$ is a northbound intersection on some cycle, $\ell(v) = -1$ if $v$ is a southbound intersection, and $\ell(v) = 0$ if $v$ is both or neither.  Directly by definitions, we have $s(n) = \sum_{v} \ell(v)$. Since any row in a colouring can have at most $\lfloor n/2 \rfloor$ nodes of colour $3$ and only nodes of colour $3$ have non-zero $\ell(v)$, we have that $|s(n)| \le n/2$.

It remains to show that $s(n)$ is odd if $n$ is odd. Assume that $n$ is odd; since the colouring $c$ is greedy, there are two adjacent nodes on row $r$ that have colours $1$ and $2$. By shifting the identifiers, we may assume that these are nodes $v_0 = (r,0)$ and $v_{n-1} = (r,n-1)$. For any node $v$ on row $r$ with colour $1$ or $2$, define the \emph{parity} of $v = (r,y_v)$ as $p(v) = y_v + c(v) \operatorname{mod} 2$. We now make the following observations:
\begin{itemize}
    \item If two nodes $u,v$ are adjacent on row $r$ with colours $1$ and $2$, and are not $v_0$ and $v_{n-1}$, they have the same parity.
    \item If two nodes $u,v$ on row $r$ with colours $1$ and $2$ are separated by a single node $w$ with $c(w) = 3$, then $u$ and $v$ have a different parity if and only if $\ell(w) \in \{ -1, 1 \}$; this follows by a simple case analysis (compare with \figureref{fig:3-col-cycles}b).
\end{itemize}
Finally, we observe that $p(v_0) \not= p(v_{n-1})$. Thus, following row $r$ from $v_0$ to $v_{n-1}$, we must have an odd number of colour $3$ nodes $v$ with $\ell(v) \in \{ -1, 1 \}$, which implies that $s(n)$ is odd.
\end{proof}

We can now solve $s$-sum coordination on directed cycles in time $T(n)$ as follows. If $n < n_0$, we gather full information about the input cycle in $n_0$ rounds; all nodes output $0$ except the one with smallest identifier, which outputs $s(n)$.  If $n \ge n_0$, we simulate $A$ on a $T(n)$-wide strip; each node looks at the middle row of the strip and outputs $\ell(v)$ as in the proof of \lemmaref{lemma:cycle-function}. Since $T(n) = o(n)$, this gives a contradiction with \theoremref{thm:q-sum-lower-bound}.

\section{\boldmath Edge colouring \texorpdfstring{$d$}{d}-dimensional grids with \texorpdfstring{$2d+1$}{2d + 1} colours}\label{app:edge-colouring}

\begin{theorem}
\label{thm:edge-colouring}
	For every fixed $d$, the complexity of edge $(2d+1)$-colouring $d$-dimensional grids is $\Theta(\log^{*} n)$.
\end{theorem}

Again, the lower bound follows from the result of \citet{linial92locality}, so it remains to show the upper bound. Moreover, we will show that this is tight in the sense that it is not possible to edge-colour the $d$-dimensional grid using $2d$ colours when $n$ is odd.

\paragraph{High-level idea.}
The general idea of the colouring is to have two exclusive colours for each dimension and to use the last remaining colour $c$ in order to colour a set of edges that cuts each row in each dimension into pieces of constant length which can then be coloured alternately by the two colours for the edges in the respective dimension.
In order to find such a set $S$ of (pairwise non-adjacent) edges, we first find a set of nodes that is able to locally choose the edges from $S$ such that the required conditions are met.

Consider an arbitrary dimension.
For each row in this dimension, find a maximal independent set of large distance and denote the union of these maximal independent sets by $M$.
Now move the nodes in $M$ on their respective rows until each node is the centre of a radius-$r$ ball (according to the $L_{\infty}$ norm, i.e., the ball is essentially a hypercube) that intersects no radius-$r$ ball from another node from $M$.
By making sure that the initial maximal independent sets are of sufficiently large distance, arbitrarily large radii $r$ can be achieved, since each node from $M$ has sufficiently large space on its row compared to the number of nodes from $M$ in its vicinity.
Repeat the whole process for each remaining dimension.

Now each node from each of the obtained $M$ colours a nearby edge (i.e., one in its radius-$r$ ball) in the row, the node was initially chosen from, with colour $c$.
By making the radii $r$ sufficiently large (depending on $d$) in the beginning, the nodes can ensure that none of these coloured edges are adjacent, since the number of radius-$r$ balls (with nodes from some of the aforementioned maximal independent sets as the centres) that intersect the nearby part of the row, from which a node chooses the to-be-coloured edge, can be bounded by a function that depends only on $d$.
Essentially, even if a node chooses the edge it wants to colour last of all choosing nodes, it always has an edge available that is not adjacent to an already coloured edge.
Moreover, the distances in the initially chosen maximal independent sets must be sufficiently large (as a function of $d$) to ensure sufficiently large radii $r$, but since $d$ is fixed they can still be chosen to be constant, and likewise the distances the nodes are moved can be bounded by a function in $d$.
Hence, the pieces obtained in each row by removing the edges of colour $c$ are of constant size and can be coloured with two colours, following our initial scheme.

\paragraph{Preliminaries.}
We use the same setting and notations as for the vertex colouring in \sectionref{app:vertex-4-colouring}. We will call a row in dimension $q$ a \emph{$q$-directional} row. As before, we call a set of nodes of $G$ a \emph{maximal independent set of distance $k$} if it is a maximal independent set in $G^{(k)}$, the $k$th power of~$G$. Furthermore, we need the following generalisation of a vertex colouring:

\begin{definition}
A vertex colouring of a $d$-dimensional grid $G$ is a \emph{colouring of $L_{\infty}$ distance $k$} if no two adjacent nodes of $G^{[k]}$ have the same colour.
\end{definition}

Note that a vertex colouring is a colouring of $L_{\infty}$ distance $2k$ if and only if for any node $u$, $B_\infty(u,k)$ contains no two nodes of the same colour.
The following lemma establishes a bound on the time it takes to find a specific colouring of a certain distance that we will need later.

\begin{lemma} \label{lem:distance-colouring}
	For every fixed $d$, there is a distributed algorithm that finds a vertex $(2k+1)^d$-colouring of $L_{\infty}$ distance $k$ of $G$ in time $O(k(\log^* n + k^d))$.
\end{lemma}

\begin{proof}
	The nodes in $G$ can simulate any distributed algorithm on $G^{[k]}$ with a multiplicative overhead of $kd$.
	Observe that a proper vertex colouring of $G^{[k]}$ induces a colouring of $L_{\infty}$ distance $k$ of $G$.
	Now, since $G^{[k]}$ has a maximum degree of $(2k+1)^d - 1$, finding a $(2k+1)^d$-colouring of $G^{[k]}$ can be done in time $O(\log^* n + (2k+1)^d)$ using the algorithm of \citet{barenboim14distributed}. Counting the simulation, total running time is $O(kd(\log^* n + (2k+1)^d))$; noting that $d$ is constant yields the desired bound.
\end{proof}

In the high-level overview of the main algorithm, we mentioned nodes that will locally choose the edges that will be coloured with the special colour that is not assigned to some specific dimension.
These nodes have to have two properties: On the one hand they should not be too far from each other in order to be able to choose a nearby edge each, such that each row in each dimension is thereby cut in sufficiently small pieces (for the later $2$-colouring of the pieces); on the other hand they should be far enough from each other such that each node has enough space to choose an edge that is not adjacent to any other chosen edge.
We formalize these considerations in the following definition:

\begin{definition} \label{def:jk-independent}
	A \emph{$j,k$-independent set w.r.t.\ dimension $q$} is a set $M$ of nodes of the grid with the following properties:
	\begin{enumerate}[label=(\arabic*),noitemsep]
		\item \label{prop:row-dist} For any node $w \notin M$ there is a node $u \in M$ in the same $q$-directional row with $\dist(u,w) \leq j$.
		\item \label{prop:cube-dist} For any two nodes $u,v \in M$ we have that $B_\infty(u,k) \cap B_\infty(v,k) = \emptyset$.
	\end{enumerate}
\end{definition}

\paragraph{\boldmath Finding a $j,k$-independent set.} In the following we describe a distributed algorithm that finds a $j,k$-independent set w.r.t.\ dimension $q$, where $j = 3(4k+1)^d$ and $k \geq 1$.

W.l.o.g.\ let $q=1$ and denote the directions belonging to dimension $1$ by \emph{west} and \emph{east} where the coordinate of dimension $1$ increases stepwise in eastern direction.
We simply use the term \emph{row} when referring to a $1$-directional row.

For each row $r$, choose a maximal independent set $M_r$ of distance $2(4k+1)^{d}$ in $r$, i.e., in the graph induced by the nodes in row $r$.
Moreover choose a (vertex) $(8k+1)^d$-colouring $c$ of $L_{\infty}$ distance $4k$ of the whole grid where the colours are chosen from $\{ 1, \dots, (8k+1)^d \}$.

Let $M$ be the union of the $M_r$ taken over all rows $r$ in the grid.
We will now transform $M$ into a $j,k$-independent set by deleting and adding nodes.
More specifically, we repeatedly delete nodes from $M$ and replace them by the respective next node in eastern direction.
When we perform such a replacement of a node $u$ by its eastern neighbour, we say that \emph{$u$ moves to the east}.
For simplicity, we denote the new node in $M$ again by $u$ and assign it the same colour $u$ had before.

The replacements take place in phases, starting with Phase $1$.
In Phase $p$ the following steps are performed:
Each node that does not have colour $p$ does nothing.
Each node $u$ of colour $p$ checks whether it is contained in $M$ and whether $B_\infty(u,2k)$ contains a second node from $M$ (i.e., a node different from $u$ itself).
If both is the case, then $u$ moves to the east and continues moving to the east until $B_\infty(u,2k)$ does not contain a node from $M$ any more, apart from $u$ itself.
Any node of colour $p$ that stops moving to the east (or did not start moving in the first step of the phase) does not start moving again, even if a node from $M$ moves into its radius-$2k$ ball.
(Recall that the radius is taken according to the $L_{\infty}$ norm.)
Each phase ends after $(4k+1)^{d} - (4k+1)$ steps upon which the next phase starts.
This concludes the description of the algorithm.

We note that any node moves to the east in at most one phase and therefore it moves at most $(4k+1)^{d} - (4k+1)$ steps to the east.
Since this is less than the distance between any two points from the same $M_r$, no node moves ``over'' another node from $M$.
Hence, the colours of the nodes that are passed by a node moving eastwards are irrelevant since only nodes from $M$ have an active role in the algorithm.
Thus, the presented algorithm is well-defined despite those colours not being specified.
Note further that we can assume that all nodes start at the same time with the different phases (and that the nodes move with the same ``speed''), by using standard synchronisation arguments.

In order to be able to show the correctness of the algorithm, we need the following lemma:

\begin{lemma} \label{lem:stop-moving}
	When a node $u$ stops moving to the east, then $B_\infty(u,2k)$ contains no node from $M$, apart from $u$ itself.
\end{lemma}

\begin{proof}
	There are two reasons a node might stop moving to the east, namely that $B_\infty(u,2k)$ contains no node from $M$, apart from $u$ itself, or that the respective phase ends.
	Since each node starts moving at the beginning of the phase (if it moves in that phase at all), it is enough to show that each node $u$ would stop moving to the east after at most $(4k+1)^{d} - (4k+1)$ steps even if the phase still continued.
	To this end, we assume for the remainder of the proof that there is no cutoff of a phase after $(4k+1)^{d} - (4k+1)$ steps, but that instead the phase ends after the last moving node stops moving.

	For any node $u$, let $Z(u)$ be the set of nodes such that $v \in Z(u)$ if and only if $v$ is in the same row as $u$ and $v_1 - u_1 = z(4k+1)$ for some integer $z \in \{ 0, \dots, (4k+1)^{d-1} - 1 \}$.
	Consider the $(4k+1)^{d-1}$ radius-$2k$ balls of the nodes in $Z(u)$.
	By the definition of $Z(u)$ these balls are lined up one after the other in eastern direction, but no pair of them intersects.
	Moreover, the union $B$ of these balls intersects exactly $(4k+1)^{d-1}$ rows of the grid and each of these rows has exactly $(4k+1)^{d}$ consecutive nodes in $B$.

	Order the nodes from $M$ by the time they stop moving, breaking ties arbitrarily.
	Note that during the whole algorithm, a node moves either once (but then possibly a number of consecutive steps) or not at all.
	In the latter case we set the point in time at which the node stops moving to $0$.
	In the case that a node moves on infinitely, we set the stopping time at $\infty$.
	Denote the ordered nodes by $u(0), u(1), \dots$ where the stopping time increases (or stays the same) with increasing argument.

	We show now by induction (on the argument of the node) that no node from $M$ moves further to the east than $(4k+1)^{d} - (4k+1)$ steps, i.e., no further than the centre of the furthest of the balls defined above.

	Consider $M$ at the point in time when $u(0)$ stops moving, or, if $u(0)$ moves on infinitely (a possibility that we cannot exclude yet), at an arbitrary point in time in the phase corresponding to the colour of $u(0)$.
	Observe that at no point in time, a node that is still moving contains another moving node in its radius-$2k$ ball.
	(The reason for this is that all (moving) nodes move synchronously and in the beginning of each phase $p$, no node of colour $p$ contains another node of colour $p$ in its radius-$2k$ ball, by the definition of our colouring $c$.)
	Hence, the current non-moving nodes in $M$ (at their current places in the grid), denoted by $M_{\textsf{static}}$, are the only nodes that can have caused $u(0)$ to move at all.
	Moreover, since $u(0)$ stops moving first, the nodes in $M_{\textsf{static}}$ are at the exact same places as they were in the beginning of phase $1$.
	In the beginning of phase $1$, any two nodes in $M_{\textsf{static}} \cup \{ u(0) \}$ in the same row have a distance of at least $2(4k+1)^{d}$, by the definition of the $M_r$.
	Since $(4k+1)^d < 2(4k+1)^{d}$, each of the $(4k+1)^{d-1} - 1$ rows intersecting $B$ contains at most one node from $M_{\textsf{static}} \cap B$, by our above observation about $B$.
	Furthermore, the row containing $u(0)$ contains no node from $(M_{\textsf{static}} \setminus \{ u(0) \} ) \cap B$.
	Hence, $B$ contains at most $(4k+1)^{d-1} - 1$ nodes from $M_{\textsf{static}} \setminus \{ u(0) \}$.
	By the pigeonhole principle, one of the $(4k+1)^{d-1}$ radius-$2k$ balls of the nodes in $Z(u(0))$ does not contain a node from $M_{\textsf{static}}\setminus \{ u(0) \}$.
	Thus, $u(0)$ stops moving when it arrives at the centre of this ball, at the latest, which yields a maximum of $((4k+1)^{d-1} - 1)(4k+1) = (4k+1)^{d} - (4k+1)$ steps taken.
	This concludes the base case of the induction.

	For the induction step, consider an arbitrary node $u(a) \in M, a \geq 1$ at the point in time it stops moving, or, if $u(a)$ moves on infinitely, at a point in time when $u(a)$ has already started to move and $u(a-1)$ has stopped moving (the induction hypothesis ensures that such a point in time actually exists).
	Define $M_{\textsf{static}}$ analogously to the definition in the base case.
	By the induction hypothesis, we can assume that each node from $M_{\textsf{static}}$ has moved at most $(4k+1)^{d} - (4k+1)$ steps to the east from its initial position.
	Thus, by the definition of the $M_r$, any two nodes in $M_{\textsf{static}}$ in the same row have a distance of at least $2(4k+1)^{d} - ((4k+1)^{d} - (4k+1)) > (4k+1)^d$ from each other (and also from the initial location of $u(a)$).
	Now the proof of the induction step follows analogously to the proof of the base case.
\end{proof}

Using Lemma \ref{lem:stop-moving}, we prove the correctness of the above algorithm and give an upper bound for its time complexity.

\begin{lemma} \label{lem:find-jk}
	Let $j = 3(4k+1)^d$, $k \geq 1$ and $1 \leq q \leq d$.
	The algorithm described above finds a $j,k$-independent set w.r.t.\ dimension $q$ in time $O(k^d \log^* n + k^{2d+1})$.
\end{lemma}

\begin{proof}
	As above, w.l.o.g.\ let $q=1$, and let $M$ denote the final set obtained by our algorithm.
	We start by showing that $M$ is indeed a $j,k$-independent set w.r.t.\ dimension $1$, by checking the properties given in Definition \ref{def:jk-independent}.

	Regarding Property \ref{prop:row-dist}, we use similar observations to the ones made in the proof of Lemma \ref{lem:stop-moving}:
	In the beginning of Phase $1$, any node has distance at most $2(4k+1)^{d}$ to some node from $M$ (as it was in the beginning of Phase $1$) in the same row, by the definition of the $M_r$.
	Since any node from $M$ moves at most $(4k+1)^{d} - (4k+1)$ steps to the east, any node in the grid has distance at most $2(4k+1)^{d} + (4k+1)^{d} - (4k+1) < 3(4k+1)^{d}$ to some node from $M$ in the same row, which proves Property \ref{prop:row-dist}.
	Property \ref{prop:cube-dist} follows from Lemma \ref{lem:stop-moving}.
	Note that if a node $u \in M$ stops moving, then no other node $v \in M$ will stop moving in $B_\infty(u,2k)$ since $v \in B_\infty(u,2k)$ is equivalent to $u \in B_\infty(v,2k)$.

	Now we examine the time complexity of our algorithm.
	Finding the $M_r$ can be done in time $O(k^d \log^* n)$, in each row in parallel, by finding a maximal independent set in the $(2(4k+1)^{d})$th power of each row.
	Finding the $(8k+1)^d$-colouring $c$ can be done in time $O(k (\log^* n + k^d))$, by Lemma \ref{lem:distance-colouring}.
	Each phase contains $O(k^{d})$ steps per node and checking the radius-$2k$ ball of a node can be done in time $O(k)$, resulting in a total time of $O(k^{2d+1})$ for the $(8k+1)^d$ phases (in a simple implementation). Thus, the total running time is $O(k^d \log^* n + k^{2d+1})$.
\end{proof}

\begin{figure}
    \centering\includegraphics[width=\columnwidth,page=5]{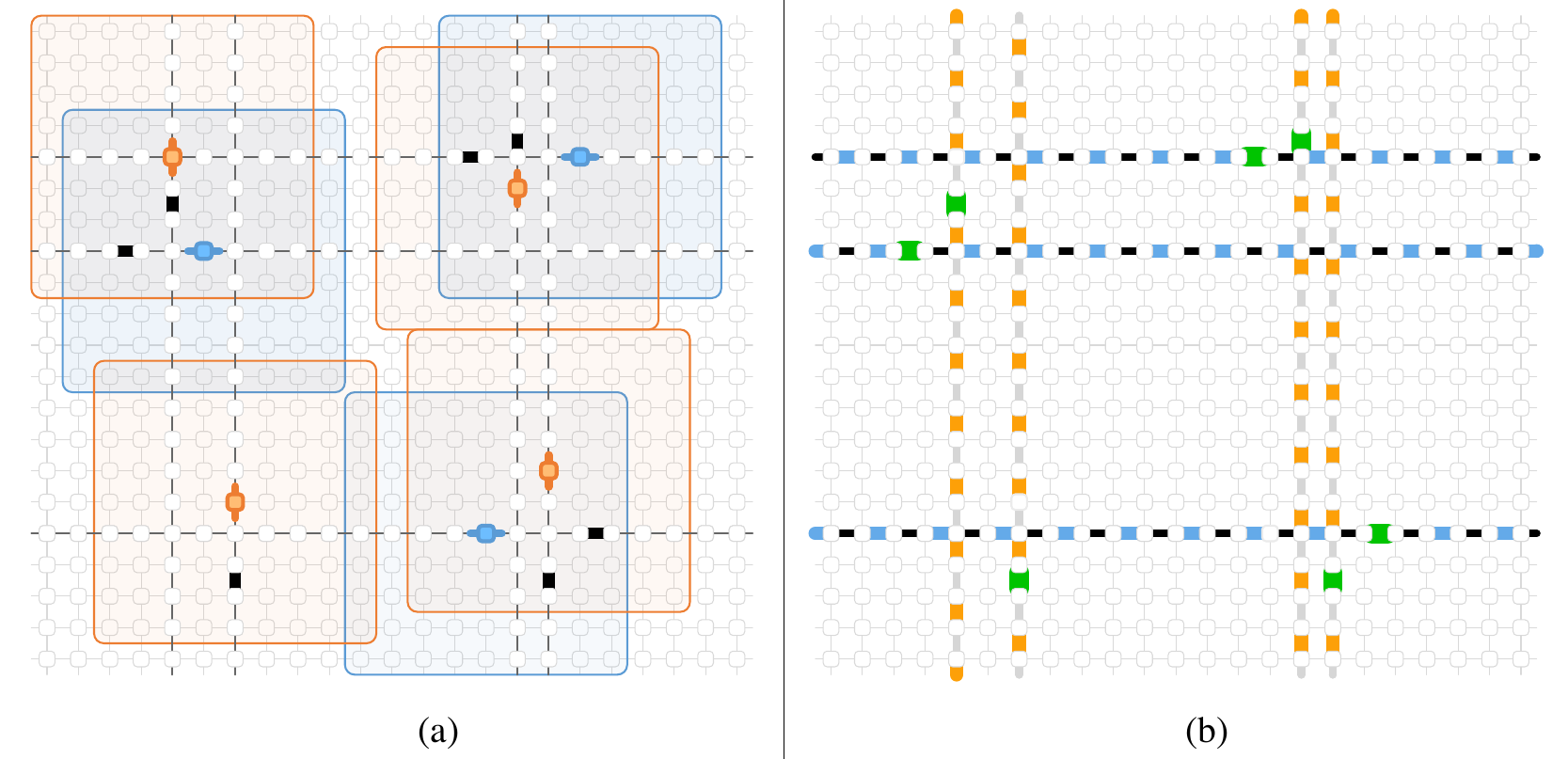}
    \caption{(a) Nodes from two $j,k$-independent sets in a $2$-dimensional grid, together with their radius-$k$ balls and the rows they are on. Each node marks an edge (shown in black) on the same row (corresponding to the dimension of the $j,k$-independent set) in its radius-$k$ ball so that marked edges are not adjacent to each other. (b) The edge colouring resulting from the edge choices made in (a); marked edges use colour $5$, shown in green. Rows with no colouring shown use colour $5$ outside the shown area.}\label{fig:edge-col}
\end{figure}

\paragraph{Upper bound for edge colouring.}
    We now proceed to describe an algorithm that finds an edge colouring with $2d+1$ colours in time $O(\log^* n)$.
	The algorithm starts by finding, for each $1 \leq q \leq d$, a $j,k$-independent set $I_q$ w.r.t.\ dimension $q$, where $j = 3(4k+1)^d$ and $k = 2d$.
	By Lemma \ref{lem:find-jk}, this is possible in time $O(\log^* n)$ since $d$ is fixed and $k$ only depends on $d$.
	Now, we again proceed in phases, starting with Phase $1$ and ending with Phase $d$.
	In Phase $p$, each node $u \in I_p$ marks an edge in $B_\infty(u,k)$ (i.e., an edge for which both of its endpoints are contained in $B_\infty(u,k)$) that is not adjacent to a previously marked edge, runs in direction of dimension $p$ and is in the same $p$-directional row as $u$ (cf.\ Figure \ref{fig:edge-col}a for an illustration in the $2$-dimensional case).

	In order to show that there is always such a non-adjacent edge in $B_\infty(u,k)$ available, consider the number of already marked edges that have at least one endpoint in $B_\infty(u,k) \cap R_p(u)$, where $R_p(u)$ denotes the set of nodes in the same $p$-directional row as $u$:
	Observe that for any arbitrary set $\mathcal B$ of pairwise disjoint radius-$k$ balls of dimension $d$ in our grid, at most two of the balls in $\mathcal B$ can intersect $B_\infty(u,k) \cap R_p(u)$ since any such intersecting ball from $\mathcal B$ must contain at least one of the two ``endpoints'' of the path $B_\infty(u,k) \cap R_p(u)$.
	Moreover, for any $1 \leq q \leq d$, the radius-$k$ balls of the nodes in $I_q$ are pairwise disjoint, by the definition of the $I_q$ and Property \ref{prop:cube-dist} of Definition \ref{def:jk-independent}.
	Hence, $B_\infty(u,k) \cap R_p(u)$ intersects at most $2(d-1)$ radius-$k$ balls of some node in some $I_q$ (apart from the ball $B_\infty(u,k)$ itself).
	Now, since an edge with at least one endpoint in $B_\infty(u,k) \cap R_p(u)$ can only be marked by a node whose radius-$k$ ball intersects $B_\infty(u,k) \cap R_p(u)$ and each of these nodes marks only one edge, there can be at most $2(d-1)$ marked edges with one endpoint in $B_\infty(u,k) \cap R_p(u)$, before $u$ marks an edge.
	Each such marked edge prevents at most two of the edges in $B_\infty(u,k)$ in the same $p$-directional row as $u$ to be marked by $u$ because of the adjacency condition.
	But since the $p$-directional row containing $u$ has $2k > 4(d-1)$ edges inside $B_\infty(u,k)$, there must be an edge left that $u$ can mark without violating any of the required conditions.
	Marking the edges as described above can be done in time $O(dk) = O(1)$.

	Now, the idea to colour the edges of the grid is simple (cf.\ Figure \ref{fig:edge-col}b for an illustration in the $2$-dimensional case):
	Each marked edge gets colour $2d+1$.
	Each remaining edge that runs in the direction of dimension $q$ gets colour $2q-1$ or $2q$.
	For that, each edge of colour $2d+1$ (or, more precisely, the endpoints of that edge) negotiates with the next edge of colour $2d+1$ in the same row in the same dimension the colouring of the in-between edges, such that the two available colours alternate.
	Observe that for any node from $I_q$ there is another node from $I_q$ in the same $q$-directional row (in both directions) with distance at most $2j+1$, by Property \ref{prop:row-dist} of Definition \ref{def:jk-independent}.
	Moreover, since each node marks an edge that is in distance at most $k$, any edge of colour $2d+1$ has a distance of at most $2k+2j+1$ to the next edge of colour $2d+1$ in the same row in the same dimension.
	Hence, the colouring of the edges can be completed (in parallel) in time $O(1)$.
	The construction of the colouring ensures that no two adjacent edges have the same colour.
	This proves the upper bound claimed in Theorem \ref{thm:edge-colouring}.

\paragraph{\boldmath Edge colouring with $2d$ colours.}
As mentioned in the beginning of this section, the bound on the number of colours given in Theorem \ref{thm:edge-colouring} is tight:

\begin{theorem}
	Let $d$ be fixed. Any $d$-dimensional grid $G_n$ with $n$ odd admits no edge $2d$-colouring.
\end{theorem}

\begin{proof}
	Let $n$ be odd and assume for a contradiction that there exists an edge $2d$-colouring of $G_n$.
	Let $c$ be one of the $2d$ colours.
	Since each node of the grid has degree $2d$, each node must have exactly one incident edge of colour $c$.
	Summing up the number of incident edges of colour $c$ over all nodes, we obtain $n^d$.
	Since this sum counts each edge of colour $c$ exactly twice, the total number of edges of colour $c$ must be $n^d/2$.
	Since $n$ is odd, $n^d/2$ is not an integer, yielding a contradiction.
\end{proof}

\section{Edge orientations}\label{app:edge-orientations}

Recall that for a set $X \subseteq \{ 0,1,2,3,4 \}$, an $X$-orientation is an orientation of the edges such that for each node $v \in V$ we have $\indeg(v) \in X$. In this section, we present an exhaustive classification for $X$-orientation problem:

\begin{theorem}
$X$-orientation problem for $2$-dimensional grids has the following complexity:
\begin{itemize}[noitemsep]
\item $\Theta(1)$ if $2 \in X$.
\item $\Theta(\log^* n)$ if $\{1,3,4\} \subseteq X$ or $\{0,1,3\} \subseteq X$.
\item Otherwise no solution exists for infinitely many $n$.
\end{itemize}
\end{theorem}

We first make the following simple observations:
\begin{itemize}[noitemsep]
    \item If $2 \in X$, the existing input orientation of the grid is a valid solution.
    \item $\{1,3,4\}$-orientation and $\{0,1,3\}$-orientation have the same complexity, as one can be obtained from the other by flipping edge directions.
\end{itemize}
The following lemmas cover the remaining cases.

\begin{lemma}
$\{1,3,4\}$-orientation has complexity $\Theta(\log^* n)$.
\end{lemma}
\begin{proof}
For the lower bound of $\Omega(\log^* n)$, notice that a constant output is not feasible solution and hence there is no constant-time solution. For the upper bound we resort to computational techniques; we can synthesise an $O(\log^* n)$-time algorithm, using techniques outlined in \sectionref{sec:synthesis} with $k = 1$.
\end{proof}

\begin{lemma}
There is no $\{1,3\}$-orientation for grids with odd $n$.
\end{lemma}
\begin{proof}
Consider grid with odd $n$ and any $\{1,3\}$-orientation of it. Since the sum of all in-degrees is equal to number of edges, being $2n^2$, it is even. Thus number of vertices with in-degree $1$ matches with parity to number of vertices with in-degree $3$, meaning that total number of vertices is even, a contradiction.
\end{proof}

\begin{theorem}
$\{0,3,4\}$-orientation problem is global on $2$-dimensional grids.
\end{theorem}
\begin{proof}
Our proof follows the steps of proof of the lower bound for vertex $3$-colouring of grids. Assume that we have an algorithm $A$ that solves the problem in $T(n) = o(n)$ rounds. We will show that such algorithm can be used to solving $q$-sum coordination problem, which by Theorem~\ref{thm:q-sum-lower-bound} leads to contradiction.

We label nodes with values of their in-degrees, that is $0$, $3$ or $4$. Observe that no two $0$ can be neighbours, similarly no two $4$ can be neighbours. When speaking of nodes labelled $3$, we will also refer to a direction of its only outgoing edge as \emph{pointing} to.

Consider two consecutive rows of vertices of the grid, $i$ and $i+1$. A row of vertical edges connecting them will be referred to as $i$-th vertical row of edges.

Consider labelling of edges from $i$-th vertical row of edges with values from $\{-1,0,+1\}$ in a following manner.
Let $u^+$ and $u^-$ be the vertices $0$ in rows $i$ or $i+1$, in the columns closest to the left and closest to the right from considered edge.
\begin{itemize}[noitemsep]
\item If there is vertex $0$ on one of the endpoint of considered edge, assign label $0$.
\item If vertices $u^+$ and $u^-$ are at odd $L_1$ distance and edge is oriented ``up'', assign label $+1$.
\item If vertices $u^+$ and $u^-$ are at odd $L_1$ distance and edge is oriented ``down'', assign label $-1$.
\item If vertices $u^+$ and $u^-$ are at even $L_1$ distance, assign label $0$.
\end{itemize}

Consider two vertex rows, $i$ and $i+1$. The gaps between nodes $0$ are at most 2 columns wide, and the only way to have exactly 2 columns gap is to put nodes labelled $3$ into $2\times 2$ square with outgoing edges forming cycle. (See Figure~\ref{fig:three-small-figures}.)

\begin{figure}[t]
\centering
\includegraphics[width=0.7\columnwidth]{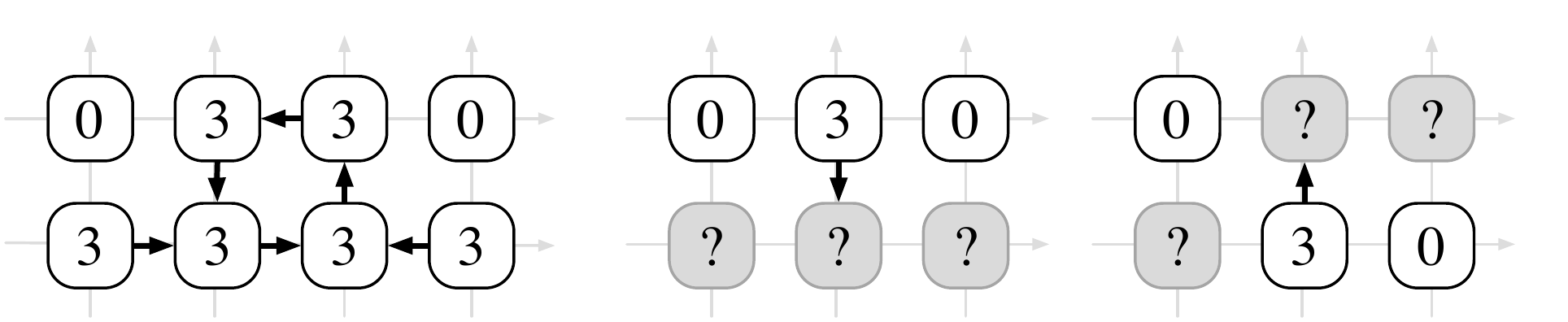}
\caption{Possible relations between vertical edges and nodes with $0$. Only in the last case the edge is labelled with non-zero value.}\label{fig:three-small-figures}
\end{figure}

Since the gaps are bounded in length, we immediately conclude that vertical edges compute their labels in at most 2 additional rounds. Let $r(i)$ be the sum of labels on the $i$-th vertical row of edges.

Consider all non-$0$ vertices and edges between them. Every vertex has out-degree at most $1$, so they form $1$-forest, with connected components being $1$-trees or trees. Observe that two vertices from separate branches of trees cannot be neighbouring, as it is not possible to orient properly edge connecting them. 
Fix one of the components as $D$. It is composed of (possibly empty) set of vertices and edges forming a directed cycle, denoted $C$, and tree-like attachments.
A border of $D$ is composed of pairwise nonadjacent $0$. However, if we consider diagonal adjacency of $0$ vertices, border is either 1 or 2 diagonally connected components: $\partial D = B_1$ or $\partial D = B_1 \cup B_2$, and each component  $B_i$ has the same parity of its vertices. We observe that only vertical edges from $C$ can contribute non-zero to any $r(i)$, as any other edge is bordered by vertices from the same $B_i$. Thus, if $B_1$ and $B_2$ have different parity, $C$ contributes $+1$ every time it crosses horizontal line ``up'' and $-1$ every time it crosses horizontal line ``down'', which is identical for all the rows (and counts the invariant of how many ``wraps around'' the cycle does on the grid).

This proves that $r(i)$ is constant for a single grid $G$, denoted $r(G)$. The proof that for two grids $G_1$ and $G_2$ of the same size, $r(G_1) = r(G_2)$, follows from adapting Lemma~\ref{lem:sum_is_const}.

Moreover, we observe that along any horizontal line, any two $u^+$ and $u^-$ contribute to the sum of labels iff they are at odd $L_1$ distance. Thus doing the full traversal and returning to the starting vertex, the parity of sum is the parity of $n$. Additionally, we observe that edges contributing non-zero to sum cannot be denser than every second edge, thus $|r(i)| \le n/2$.

The claim now follows by a straightforward application of Theorem~\ref{thm:q-sum-lower-bound} to $r(G)$.
\end{proof}

\section{Discussion and open questions}\label{sec:discussion}

\paragraph{Randomised complexity.} \citet{chang16exponential} showed that the randomised complexity of any \lcl{} on instances of size $n$ is at least its deterministic complexity on instances of size $\sqrt{\log n}$. This, combined with our \theoremref{thm:speedup}, implies that there are no \lcl{} problems with randomised complexity between $\omega(\log^* n)$ and $o(\sqrt{\log n})$ on the grid. Whether problems with randomised complexity $O(\sqrt{\log n})$ exist is left as an open question.

\paragraph{High-dimensional grids.} As mentioned in the beginning of the introduction, we can also consider the setting of $d$-dimensional (hypertoroidal oriented) grids with $n^d$ nodes. The complexity results extend to this setting: the classification theorem and the undecidability of classification hold for $d$-dimensional grids, and as noted before, the vertex and edge colouring results generalise. The techniques used in the synthesis algorithm also generalise to $d$-dimensional grids. However, we have not yet implemented the synthesis beyond $d=2$, and we expect that the increased size of the search space may make the synthesis less feasible.

\paragraph{Bounded growth graphs.} The proof of \theoremref{thm:speedup} intrinsically exploits the fact that the size of a neighbourhood $N_r(v)$ grows quadratically in $r$, and thus any algorithm with running time $T(n) = o(n)$ cannot see all $n^2$ nodes of the graph for large $n$. We show that this is not a phenomenon restricted to grids: for any class of graphs with limited neighbourhood growth rate, we get a large complexity gap. See Appendix~\ref{sec:general-speedup} for the precise statement and the proof.

\paragraph{Sublinear problems on general graphs.} Finally, we use techniques inspired by grid graphs to expand our understanding of the complexity landscape of \lcl{} problems on general bounded-degree graphs. Recall that in general we know that the lower end of the complexity landscape is sparse: for deterministic algorithms, there is nothing between the classes $O(1)$, $\Theta(\log^* n)$, and $\Theta(\log n)$. There are also obviously problems of complexity $\Theta(n)$, but the gap between $\Theta(\log n)$ and $\Theta(n)$ is largely unexplored. In Appendix~\ref{app:root} we show how to engineer an \lcl{} problem with a complexity of precisely $\Theta(\sqrt{n})$ in general bounded-degree graphs; subsequently, \citet{hierarchy} have given a more general result showing that problems of complexity $\Theta(n^{1/k})$ exists for any integer $k \ge 2$ even when restricted to bounded-degree trees.

\section*{Acknowledgements}

We would like to thank Orr Fischer for many discussions related to these research questions.
This work was supported in part by the Academy of Finland, Grants 285721 and 289002.

\DeclareUrlCommand{\Doi}{\urlstyle{same}}
\renewcommand{\doi}[1]{\href{http://dx.doi.org/#1}{\footnotesize\sf doi:\Doi{#1}}}

\bibliographystyle{plainnat}
\bibliography{grid-lcl}

\clearpage
\appendix
\section{Appendix}

\subsection{Generating tiles}\label{app:tiles}

Consider a graph $G$ with a maximal independent set $I$. A \emph{tile} of
$(G,I)$ is a pair $(G',I')$, where $G'$ is an induced subgraph
of $G$ and $I' = V(G') \cap I$. Observe that the property of being a tile is
\emph{hereditary}: If $(G',I')$ is a tile, $G''$ is an induced
subgraph of $G'$, and $I'' = V(G'') \cap I'$, then $(G'',I'')$ is a tile of the original
graph $G$. Consequently, one may construct tiles of a graph through a sequence of
induced subgraphs. To present the algorithm for one step in such a sequence we
need to define the concept of closed neighbourhood.

The \emph{closed neighbourhood} of a vertex
$v \in V(G)$ consists of $v$ and the vertices adjacent to $v$ and
is denoted by $N_G[v]$. For a set of vertices $V' \subseteq V(G)$, we further define
$N_G[V'] := \bigcup_{v \in V'}N_G[v]$.
In the sequel, unless otherwise
mentioned, we assume that we are dealing with the graph $G$ so that
$G'$ is an induced subgraph of $G$ and $G''$ is an
induced subgraph of $G'$ (as well as $G$, obviously).

We now want to extend tiles $(G'',I'')$ to tiles $(G',I')$
in all possible ways (the unknown
is $I'$). Let $V_d = (V(G') \setminus V(G'')) \setminus N_G[I'']$.
For each independent set $I_d$ of
$V_d$, $(G',I''\cup I_d)$ is a candidate to be tile and has to be checked.
This can be done as follows. Let $V_u = V(G') \setminus N_G[I''\cup I_d]$. If $V_u = \emptyset$,
then we have a tile since any independent set can be extended to a maximal
independent set and none of the additional vertices could come from $V(G')$.

If $V_u \neq \emptyset$, then we have a tile if and only if there is an
independent set $I_n$ in $(V(G) \setminus V(G')) \setminus N_G[I'' \cup I_d])$ such that
$V_u \subseteq N_G[I_n]$. We now form sets
$S_v = (N_G[v] \setminus V(G')) \setminus N_G[I'' \cup I_d])$ for each
$v \in V_u$ and the computational problem is to find an independent set
that intersects each of $S_v$. This resembles variants of the set cover (hitting set)
problem, and corresponding instances may be solved using a SAT solver or
by implementing a tailored backtrack search, e.g., along the lines of \citet{knuth00dancinglinks}.

For (powers of) toroidal grid graphs we consider ``rectangular'' tiles with
$a \times b$ vertices and sequences such as
$0 \times b \rightarrow 1 \times b \rightarrow 2 \times b \rightarrow \cdots \rightarrow a
\times b$.

In the construction of tiles, one could further make use of symmetries (automorphism
groups) to achieve some speed-up. However, for tiles of the types mentioned above, the orders
of the groups are only 4 (rectangular case) and 8 (square case), and the basic algorithm
is already fast enough to handle the instances considered in this work.

\subsection{Speed-up on graphs of bounded growth}\label{sec:general-speedup}

Let $f \colon \N \to \N$ be a strictly increasing function, and let $\mathcal{G}$ be a class of graphs. Recall that $\mathcal{G}$ is \emph{$f$-growth-bounded} if for every $G \in \mathcal{G}$ and $v \in V(G)$, we have that $| N_r(v) | \le f(r)$. Moreover, we say that $\mathcal{G}$ is \emph{neighbourhood-hereditary} if there is a constant $C$ such that for any graph $G \in \mathcal{G}$, vertex $v \in V(G)$ and constant $r$, for any $k\ge C|N_r(v)|$ there is a graph $G' \in \mathcal{G}$ such that $|V(G')| = k$ and $G[N_r(v)]$ is isomorphic to a induced subgraph of $G'$.

\begin{lemma}\label{lemma:speedup-bounded-growth}
Let $\mathcal{G}$ be a neighbourhood-hereditary $f$-growth-bounded family of graphs with constant maximum degree $\Delta$, where $f(n) = \omega(n)$. If \lcl{} problem $P$ can be solved deterministically on $\mathcal{G}$ in $T(n) = o\bigl(f^{-1}(n)\bigr)$ rounds, then $P$ can be solved deterministically on $\mathcal{G}$ in $O(\log^* n)$ rounds.
\end{lemma}

\begin{proof}
Denote the radius parameter of the \lcl{} problem $P$ with $r$, and assume that there exists an algorithm $A$ for $P$ with complexity $T(n) = o\bigl(f^{-1}(n)\bigr)$. We show that there exists an algorithm $A'$ for $P$ that runs in $O(\log^* n)$ rounds on $\mathcal{G}$.

First, we fix a constant $k$ such that $f(2T(k)+3) < k/C$, where $C$ is as in the definition of the neighbourhood-hereditary graph class; such $k$ exists since we have $T(n) = o\bigl(f^{-1}(n)\bigr)$. The algorithm $A'$ now functions as follows on an input graph $G \in \mathcal{G}$:
\begin{enumerate}[noitemsep]
	\item Find a distance-$(2T(k)+3)$ colouring with $f(2T(k)+3)+1 \le k$ colours. This can be done in $O(\log^* n)$ rounds by simulating a $(\Delta + 1)$-colouring algorithm in the power graph $G^{(2T(k)+3)}$, since the maximum degree in $G^{(2T(k)+3)}$ is at most $k$.
	\item Simulate $A$ on $G$ with implicit assumption that the instance size is $k$, using the colours given by (1) as unique identifiers; as the simulation runs in $T(k)$ rounds, nodes will not see any duplicate colours.
\end{enumerate}
Now consider any node $v \in V(G)$; we want to show that the labelling given to $N(v)$ by $A'$ is valid. Since $\mathcal{G}$ is neighbourhood-hereditary class, there is a graph $G' \in \mathcal{G}$ with $|V(G')| = k$ such that $N_{2T(k)+3}(v)$ is isomorphic to an induced subgraph of $G'$. Moreover, since no colour given by (1) occurs twice in $N_{2T(k)+3}(v)$, this colouring can be extended to a valid assignment of unique identifiers on $G'$. Since $A$ produces a valid output on $G'$, the output of $A'$ on $N(v)$ is also valid.
\end{proof}

\subsection{Sublinear complexity problems on general graphs}\label{app:root}

In this paper we have considered $n \times n$ grids. In this setting a global problem has $\Omega(\sqrt N)$ complexity where $N=n^2$ is the size of the input.
We define an \lcl{} problem with complexity $\Theta(\sqrt{n})$ where the input is a graph $G$ on $n$ vertices without any restrictions. The basic idea is that if $G$ is a grid, we force the corners to coordinate, and otherwise we allow the nodes to have any output. Before we define the problem, we introduce the following terms to describe nodes with different local neighbourhoods.
Any node whose $O(1)$ radius neighbourhood is not isomorphic to the neighbourhood of some node in a grid is said to be a {\em broken} node. Any other node is a {\em corner} node if it has degree $2$ and an {\em internal} node otherwise.

\paragraph{The corner coordination problem:}

\begin{itemize}
    \item If there are no corner nodes, then nodes can output anything.
    \item Otherwise, nodes must direct some (or possibly none) of their incident edges according to the following rules:
    \begin{enumerate}
        \item The set of directed edges forms a set of directed pseudotrees: each node must have at most one outgoing edge in each tree.
        \item The pseudotrees have a consistent orientation: a path in one of the pseudotrees can cross each row and column at most once.
        \item Only corner nodes can be roots or leaves of the pseudotrees.
        \item Pseudotrees can only meet at corners or broken nodes.
        \item Each corner must be the root or leaf of at least one pseudotree.
    \end{enumerate}
\end{itemize}

In order to make this a locally checkable labelling problem, the output of a node $v$ should be a (possibly empty) set of labels for its incident edges, which must include an indication of the forbidden rows and columns for any pseudotree containing the labelled edge. This can be achieved in the following manner. Suppose $v$ needs to direct the edge $e=vw$ toward $w$. Then $v$ can include in the label the identifier of one of its neighbours $v'$ that is a forbidden neighbour of a successor of $w$ in the pseudotree. In other words, if $w$ wants to direct the edge $e'=wx$ towards $x$, then $x$ cannot be adjacent to $u$. Furthermore, if $w$ does direct $e'$ toward $x$, it will include a the identifier of $w'$ in the label and this must be consistent with the choice of $v'$: $w'$ should be adjacent to $v'$ or it should be adjacent to $v$. This is certainly locally checkable, but we need the set of labels to be of constant size. We can derive a port numbering from the unique identifiers of the neighbourhood of $v$. Instead of the identifier, $v$ can include the port number of $v'$ in the label of $e$. The degree of $v$ is at most $4$ and so the set of labels is of constant size.

\begin{theorem}
The corner coordination problem has complexity $\Theta(\sqrt{n})$.
\end{theorem}

\begin{proof}
First we prove a lower bound. Suppose there exists an algorithm $A$ for the corner coordination problem that runs in time $T(n)=o(\sqrt{n})$. Let $G$ be a $2$-dimensional grid on $N=m^2$ vertices. We refer to the nodes of $G$ by coordinates $i,j$ in the obvious way so that $v_{0,0}$ has degree 2. This is a convention to aid our discussion; the nodes are not aware of these coordinates. Consider the set of pseudotrees in the output of $A(G)$. Since every corner node must be the root or leaf of at least one pseudotree, there must be a pseudotree $T$ that consists of a path along one side of the grid. Without loss of generality $T$ is the directed path $(v_{0,0},v_{0,1},\ldots,v_{0,m})$. We obtain a new input graph $G'$ from $G$ by taking the ball $B_\infty (v_{0,\frac{m}{2}},\epsilon m)$ for some sufficiently small constant $\epsilon$ and rotating it about $v_{0,\frac{m}{2}}$. Now consider the output of $A(G')$. Since the $T(n)$-radius neighbourhood of $v_{0,0}$ in $G$ is isomorphic to the $T(n)$-radius neighbourhood of $v_{0,0}$ in $G'$, there is a tree $T_{0,0}$ in the output of $A(G')$ whose root is $v_{0,0}$. Similarly, there is a tree $T_{0,m}$ whose leaf is $v_{0,m}$. There is also a pseudotree $T'$ with a path going through $v_{0,\frac{m}{2}}$ but in $G'$ this path goes the ``wrong" way in the sense that it points {\em towards} $v_{0,0}$. This forces $T_{0,0}$ and $T_{0,m}$ to be different trees, since $T_{0,0}$ must eventually include a vertex $v_{1,j}$ and therefore its leaf can not be $T_{0,m}$. Furthermore, since the first edge of $T_{0,0}$ is $v_{0,0}v_{0,1}$, the leaf of $T_{0,0}$ cannot be $v_{m,0}$. So the leaf of $T_{0,0}$ must be $T_{m,m}$ and by a similar argument, the root of $T_{0,m}$ is $T_{m,0}$. But this is a contradiction as the pseudotrees cannot cross.

Now we show that the problem can be solved in $2\sqrt{n}$ rounds. It is enough to show that in $2\sqrt{n}$ rounds, a corner node $v$ sees a corner node or a broken node. Suppose that $v$ has not seen a corner or broken node after $r$ rounds.

\begin{proposition}
The $r$-radius neighbourhood of a corner node that has not seen a corner or broken node contains $\binom{r+2}{2}$ nodes.

\end{proposition}

The number of nodes at distance exactly $k$ from the corner is at most the number of ordered pairs of non negative integers that sum to $k$, which is $k+1$. It is well known that $\sum_{k=0}^{r} {k+1} = \binom{r+2}{2}$. Now when $r=2\sqrt{n}$, the number of vertices that $v$ has seen is greater than $n$. So in $2\sqrt{n}$ rounds, $v$ must see a corner node or a broken node. This completes the proof.
\end{proof}

\end{document}